\documentclass{article}
\usepackage{amsmath,amssymb,amsthm}
\usepackage[dvipdfm]{graphicx}

\newtheorem{theorem}{Theorem}[section]

\newtheorem{proposition}[theorem]{Proposition}
\newtheorem{lemma}[theorem]{Lemma}
\newtheorem{remark}[theorem]{Remark}

\newcommand{\arc}[1]{\mathrm{arc}(#1)}

\begin{document}

\title{Long-time asymptotics for the defocusing \\
integrable 
discrete nonlinear Schr\"odinger equation}

\author{Hideshi  Yamane}%

\maketitle

\abstract{
We investigate the long-time asymptotics for the 
defocusing integrable discrete nonlinear Schr\"odinger equation 
of Ablowitz-Ladik 
by means of the inverse scattering transform and the  Deift-Zhou nonlinear steepest descent method. 
The leading term is a sum of two terms that oscillate 
with decay of order $t^{-1/2}$. 

 
}

\tableofcontents
\bigskip
\hrule
\markboth{Hideshi Yamane}{Defocusing integrable 
discrete nonlinear Schr\"odinger equation}

\section{Introduction}

In this article we study the long-time behavior of the defocusing integrable 
discrete nonlinear Schr\"odinger equation (IDNLS) introduced by  
Ablowitz and Ladik (\cite{AblowitzLadik, AblowitzLadik2, APT}) on the doubly infinite lattice 
(i.e. $n\in\mathbb{Z}$)
\begin{equation}
   i \frac{ d}{ dt}R_n +
 (R_{n+1}-2R_n+R_{n-1})-|R_n|^2 (R_{n+1}+R_{n-1})=0.
\label{eq:IDNLS}
\end{equation}
It is a discrete version of the defocusing nonlinear Schr\"odinger 
equation (NLS) 
\begin{equation}
\label{eq:NLS}
  i u_t+u_{xx}-2u|u|^2=0  \mbox{ or }  
  i v_t-v_{xx}+2v|v|^2=0  \qquad (v=\bar u). 
\end{equation}
Although there are other ways to discretize \eqref{eq:NLS}, 
we have chosen   \eqref{eq:IDNLS}  because of the 
striking fact that it is \textit{integrable}: it can be solved by 
\textit{the inverse scattering transform} (IST). 
Here we employ the Riemann-Hilbert formalism of IST, rather than that based on 
integral equations. 
Knowledge of the IDNLS can give insight for the non-integrable versions, 
especially when one is interested in  asymptotics.

Significant works have been done on the long-time 
behavior of   integrable equations, pioneers being 
\cite{AN, Ma, ZMa}.  
The epoch-making work by Deift and Zhou in \cite{DZ} on the MKdV equation 
developed the inverse scattering technique and established 
\textit{the nonlinear steepest descent method}. 
It was used to  study   
the defocusing nonlinear Schr\"odinger equation 
by Deift, Its and Zhou in \cite{Important} 
and the Toda lattice 
in \cite{Kamvissis, KTsolitonregion, KTrevisited}. 
A detailed bibliography about the focusing/defocusing nonlinear  
Schr\"odinger equations on the (half-)line or  an interval 
is  found in \cite{Fokasbook}.

Following the above mentioned results, we employ 
  the Deift-Zhou  nonlinear steepest 
descent method and obtain the long-time asymptotics of  \eqref{eq:IDNLS}. 
Roughly speaking, the result is as follows.  
(See  \S\ref{sec:Main result} for details.)
If $|n/t|<2$, there exist $C_j=C_j(n/t)\in\mathbb{C}$ and 
$p_j=p_j(n/t),  q_j=q_j(n/t)\in\mathbb{R}$ ($j=1, 2$)
depending only on the ratio $n/t$ such that 
\begin{equation}
\label{eq:introstatement}
 R_n(t)=\sum_{j=1}^2 C_j t^{-1/2}  e^{- i (p_j t+q_j \log t)}
 +O(t^{-1}\log t)
 \quad\mbox{as }t\to\infty. 
\end{equation}
The quantities $C_j, p_j$  and $q_j$ are 
defined in terms of \textit{the reflection coefficient}  
that appears in the inverse scattering formalism. 	
The behavior of each term in the sum  is \textit{decaying oscillation} 
of order $t^{-1/2}$. 
Notice that in the case of the continuous defocusing NLS 
\eqref{eq:NLS}, the asymptotic 
behavior is expressed by a single term, not a sum, 
with decaying oscillation of order $t^{-1/2}$. 
Notice that the defocusing NLS and IDNLS are 
without  solitons vanishing rapidly at infinity. 
(Dark solitons do not vanish at infinity.)

In  \cite{Michor}, 
Michor studied the spatial asymptotics ($n\to\infty$, $t$: fixed) 
of solutions of  \eqref{eq:IDNLS} (and  its generalization called 
the Ablowitz-Ladik hierarchy). 
She proved that the leading term is $a/n^\delta \, (\delta\ge 0)$ 
in sharp contrast to \eqref{eq:introstatement} 
under a certain assumption on the  initial value. 
A natural remaining problem is to determine the asymptotics in 
$|n/t| \ge 2$, which will be a subject of future research.

Another interesting problem is to find the long-time asymptotics  for 
the {focusing} IDNLS. It is more difficult than the defocusing one 
because the associated   Riemann-Hilbert problems may have poles corresponding to 
solitons vanishing at infinity.

\begin{remark}
The term $-2R_n$ in \eqref{eq:IDNLS} can be removed by 
a simple transformation $e^{2it}R_n(t)=\tilde R_n(t)$ and 
some authors prefer this formulation. 
\end{remark}

\section{Inverse scattering transform for the defocusing IDNLS}
In this section we explain some known facts about 
inverse scattering transform for the defocusing IDNLS 
following \cite[Chap. 3]{APT}, which is a refined version of 
\cite{AblowitzLadik, AblowitzLadik2}.

First we discuss unique solvability of the 
Cauchy problem for \eqref{eq:IDNLS}. 

\begin{proposition}
Assume that the initial value $R(0)=\{R_n(0)\}_{n\in\mathbb{Z}}$ satisfies 
\begin{align}
 &\|R(0)\|_1=\sum_{ n=-\infty}^{\infty}|R_n(0)|<\infty,
 \label{eq:l1}
 \\
  &\|R(0)\|_\infty=\sup_n |R_n(0)|<1 \quad\mbox{\upshape (smallness condition)}.
 \label{eq:smallnesscondition}
\end{align}
Then \eqref{eq:IDNLS} has a unique solution 
in $\ell^1=\{\{c_n\}_{n=-\infty}^\infty\colon \sum |c_n|<\infty\}$ for $0\le t<\infty$. 
\end{proposition}

\begin{proof}
We can regard \eqref{eq:IDNLS}  as an ODE in the Banach space 
$\ell^1\subset \ell^\infty$. 
First we solve it in 
$\ell^\infty$ in view of 
\eqref{eq:smallnesscondition}.  
Set $c_{-\infty}=\prod_{n=-\infty}^\infty (1-|R_n|^2)>0$, 
$\rho=(1-c_{-\infty})^{1/2}$.   
Since $1-|R_n(0)|^2\ge c_{-\infty}$ for each $n$, we have $\|R(0)\|_\infty \le \rho$. 
Set 
\(
  B:=\{R=  \{R_n\}\in\ell^\infty; 
            \|R -R (0)\|_\infty \le \rho 
       \}
\). 
Since the right-hand side 
is Lipschitz continuous and bounded if $R=\{R_n\}\in B$,  
\eqref{eq:IDNLS}  can be solved  in $B$ locally in time,  
say up to $t=t_1=t_1(\rho)$. 
By a standard argument about ODEs in a Banach space,  
$t_1$ is determined by $\rho$ only and is 
independent of $R(0)=\{R_n(0)\}$ as long as $ \sup_n |R_n(0)|\le \rho$. 
Since  it is known that $c_{-\infty}$ and $\rho$ are conserved quantities, we have 
\(\|R(t)\|_\infty=
 \sup_n |R_n(t)|\le \rho
\) for $0\le t<t_1$.  
Then we solve  \eqref{eq:IDNLS} again with the initial value at  
$t=t_1/2$. The solution can be extended up to $t=3t_1/2$. 
Repetition of this process enables us to extend the 
solution $\{R_n(t)\} \in \ell^\infty$ up to $t=\infty$ 
and it satisfies \( \sup_n |R_n(t)|\le \rho\)  for $0\le t<\infty$.  
We have 
$\|\frac{d}{dt}R_n(t)\|_1\le \textrm{const.} \|R_n(t)\|_1$. Therefore 
$ \|R_n(t)\|_1$ grows at most exponentially and $\{R_n(t)\}$ belongs  
to  $\ell^1$ for any $t<\infty$. 
\end{proof}

Next we explain a concrete representation formula of the solution 
based on  inverse scattering transform. 
Let us introduce the associated Ablowitz-Ladik scattering problem 
(a difference equation, not a differential equation) 
\begin{equation}
 X_{n+1}=
 \begin{bmatrix}
   z & \overline{R}_n  \\ R_n & z^{-1} 
 \end{bmatrix}
 X_n. 
\label{eq:Ablowitz-Ladik}
\end{equation}
It has no discrete eigenvalues (\cite[p.66]{APT}, 
\cite[Appendix]{Abetal}) 
and  \eqref{eq:IDNLS} has no soliton solution vanishing at infinity.

\begin{remark}
We adopt the formulation of \cite{APT}. If one follows that of 
\cite{Important}, the matrix on the right-hand side of 
\eqref{eq:Ablowitz-Ladik} 
should be replaced by 
\(
 \begin{bmatrix}
   z &  R_n  \\ \overline{R}_n & z^{-1} 
 \end{bmatrix}
\) and it leads to minor changes at many places. 
\end{remark}

The time-dependence equation is 
\begin{equation}
 \frac{\, d}{\, d t}X_n
 =
 \begin{bmatrix} 
    i R_{n-1}\overline R_n -\frac{ i}{2}(z-z^{-1})^2 
   & - i (z \overline R_n-z^{-1}\overline R_{n-1})
   \\
    i (z^{-1} R_n- z R_{n-1}) 
   & - i R_n \overline R_{n-1} +\frac{ i}{2}(z-z^{-1})^2
 \end{bmatrix}
 X_n 
 \label{eq:timedependence}
\end{equation}
and  \eqref{eq:IDNLS}  is equivalent to \textit{the  
compatibility condition} 
\(
 \frac{\, d}{\, d t}X_{n+1}=(\frac{\, d}{\, d t}X_m)_{m=n+1}
\)    if we substitute \eqref{eq:Ablowitz-Ladik} and 
 \eqref{eq:timedependence}  into the left and right-hand sides 
 respectively.

The conditions \eqref{eq:l1} and \eqref{eq:smallnesscondition} 
are preserved for $t<\infty$. 
We can construct eigenfunctions (\cite[pp.49-56]{APT}) 
satisfying \eqref{eq:Ablowitz-Ladik} for  any fixed $t$.  
More specifically, one can define the eigenfunctions (depending on $t$)
\(
  \phi_n(z, t), \psi_n(z, t)   \in 
  \mathcal{O}(|z|>1) \cap \mathcal{C}^0 (|z|\ge 1)
\)  and 
\(
  \bar\psi_n(z, t)   \in 
  \mathcal{O}(|z|<1) \cap \mathcal{C}^0 (|z|\le 1)
\) 
such that 
\begin{align}
 &\phi_n(z, t)\sim z^n 
 \begin{bmatrix} 1 \\ 0 \end{bmatrix} 
 & \mbox{ as } n\to-\infty,
 \\
 & \psi_n(z, t)\sim z^{-n} 
 \begin{bmatrix} 0 \\ 1 \end{bmatrix}, 
 \quad
 \bar\psi_n(z, t)\sim z^n 
 \begin{bmatrix} 1 \\ 0 \end{bmatrix} 
 & \mbox{ as } n\to \infty.
\end{align}
On the circle $C\colon |z|=1$,  
there exist unique functions 
$a(z, t)$ and $b(z, t)$ for which 
\begin{equation}
 \phi_n(z, t)=b(z, t)\psi_n(z, t)+a(z, t)\bar\psi_n(z, t)
\end{equation}
holds. They can be represented as Wronskians of the 
eigenfunctions. The characterization equation 
\begin{equation}
 |a(z, t)|^2-|b(z, t)|^2=c_{-\infty}>0
\label{eq:characterization}
\end{equation} 
implies $a(z, t)\ne 0$. Hence one can define the 
\textit{reflection coefficient} 
\footnote{It is denoted by $\rho$ in the notation of  \cite{APT}. }
\begin{equation}
 r(z, t)=\frac{b(z, t)}{a(z, t)}. 
\end{equation}
In our notation,  $r(z)$ is short for $r(z, 0)$, not for $r(z, t)$. 
It has the property 
$r(-z, t)=-r(z, t), 0\le |r(z, t)|<1$, the latter being a consequence of 
\eqref{eq:characterization}. 
If $\{R_n(0)\}$ is rapidly decreasing in the sense that 
\begin{equation}
 \sum_{n\in\mathbb{Z}} |n|^s |R_n(0)|<\infty
 \text{ for any } s\in\mathbb{N}, 
\label{eq:rapidlydecreasing}
\end{equation}
then 
$\phi_n, \psi_n$ and $\bar\psi_n$ are smooth on $C$, hence 
so are $a, b$ and $r$.

The time evolution of $r(z)$ according to \eqref{eq:timedependence} 
is given by 
\begin{equation}
 r(z, t)=r(z)\exp\left(it (z-z^{-1})^2\right)
 =r(z)\exp\left(it (z-\bar z)^2\right), 
\end{equation}
where $r(z)=r(z, 0)$. 

Let us formulate the  following Riemann-Hilbert problem:
\begin{align}
   & m_+(z)=m_-(z) v(z) \;\mbox{\,on\,}\;   C\colon |z|=1, 
  \label{eq:originalRHP1}
  \\
  & m(z)\to I   \;\mbox{\,as\,} \; z\to \infty,   
  \label{eq:originalRHP2}
  \\
  &  
  v(z)=v(z, t)=
	   \begin{bmatrix}
	    1-|r(z, t)|^2 \;&\; -z^{2n}\bar{r}(z, t)
	    \\
	    z^{-2n}r(z, t) \;&\; 1
	   \end{bmatrix}     	   \nonumber
  \\
  &	\hspace{1.7em}
  	   = e^{ -\frac{ it}{2}(z-z^{-1})^2 \mathrm{ad\,}\sigma_3}
	   \begin{bmatrix}
	    1-|r(z)|^2 \;&\; -z^{2n}\bar r(z)
	    \\
	    z^{-2n}r(z) \;&\; 1
	   \end{bmatrix}.   
   \label{eq:originalRHP3}
\end{align}
Here $m_+$ and $m_-$  are the boundary values from the 
\textit{outside} and \textit{inside} of $C$ respectively 
of the unknown matrix-valued  analytic function $m(z)=m(z; n, t)$ 
in $|z|\ne 1$. 
We employ the usual notation 
$\sigma_3=\mathrm{diag\,}(1, -1)$, 
$a^{\mathrm{ad\,}\sigma_3}Q=a^{\sigma_3} Q a^{-\sigma_3}$ 
($a$: a scalar, $Q$: a $2\times2$ matrix). 
The inconsistency with the usual counterclockwise orientation 
(the inside being the plus side) is irrelevant 
because later we will choose different orientations on different parts 
of the circle for a technical reason. 

The uniqueness of the solution to the problem above is derived 
by a Liouville argument. If $m$ and $m'$ are solutions, then 
$mm'^{-1}$ is equal to $I$ because it is entire and tends to 
$I$ as $z\to\infty$. The existence of the solution follows from the 
Fredholm argument in \cite{Zhou}.

The solution $\{R_n\}=\{R_n(t)\}$ to \eqref{eq:IDNLS} can be obtained 
from the $(2, 1)$-component of $m(z)$ 
by the   reconstruction formula (\cite[p.69]{APT}) 
\(
   m(z)_{21} = -z R_n(t) +O(z^2)
\), i.e., 
\begin{equation}
 R_n(t)=-\lim_{z\to 0} \frac{1}{z}m(z)_{21} 
 =-\left.\frac{\, d}{\, d z} m(z)_{21} \right|_{z=0}
 \,  (\mathrm{NB}\colon z\to 0, \mbox{not} \,\infty). 
 \label{eq:Rnreconstruction}
\end{equation}

Summing up, the inverse scattering procedure is as follows: 
\begin{itemize}
\item The initial value $\{R_n(0)\}$ determines 
  $r(z)=r(z, 0)$.
\item
 $m(z)$ is determined by $r(z)$ as the solution to  
 \eqref{eq:originalRHP1}-\eqref{eq:originalRHP3}. 
\item
 $\{R_n(t)\}$ is derived from $m(z)$ by \eqref{eq:Rnreconstruction}. 
\end{itemize}
Set 
\[
 \varphi=\varphi(z)=\varphi(z; n, t)=\frac{1}{2} i t (z-z^{-1})^2 
 - n\log z 
\] 
so that the jump matrix $v$ in \eqref{eq:originalRHP1} is given by 
\begin{equation}
 v=v(z)=\rm e^{-\varphi \mathrm{ad}\sigma_3}
 \begin{bmatrix}
 1-|r(z)|^2 & \quad -\bar r(z) \\ r(z) & \quad  1
 \end{bmatrix}.
\end{equation}
This representation is useful in that the saddle points of $\varphi$ 
play important roles in the method of nonlinear steepest descent.

\begin{remark}
Explicit expressions of some solutions are discussed in \cite{Abetal, CCX}. 
\end{remark}

\section{Main result}
\label{sec:Main result}
In the following, we will deal with the asymptotic behavior of $R_n(t)$ 
as $t\to\infty$ in the region defined by
\begin{equation}
  |n| \le (2-V_0)t, \quad
 V_0 \;\mbox{is a constant with\;} 
 0< V_0 <2.
 \label{eq:region}
\end{equation}

We have $\, d \varphi/\, d z=0$ if and only if 
$z=S_j\in C\,(j=1, 2, 3, 4)$,  where 
\begin{align}
 & S_1=  e^{-\pi  i/4}A, \, S_2= e^{-\pi  i/4}\bar{A},  \, 
 S_3=-S_1, \, S_4=-S_2, 
 \\
 &  A=2^{-1}\bigl(\sqrt{2+n/t\,}- i \sqrt{2-n/t\,}\,\bigr), 
\end{align}
and we set $S_{j\pm 4}=S_j$ by convention.

Set 
\begin{align}
 &\delta(0)=
 \exp\left(
  \frac{-1}{\pi  i}
   \int_{S_1}^{S_2} 
  \log (1-|r(\tau)|^2) \frac{\, d \tau}{\tau}
  \right), 
 \\
 &
 \beta_1
 =\frac{- e^{\pi  i/4}A}{2(4t^2-n^2)^{1/4}}, \;
   \beta_2=\frac{ e^{\pi  i/4}\bar A}{2(4t^2-n^2)^{1/4}},  
 \\
 &
 D_1=\frac{- i A}{2(4t^2-n^2)^{1/4}(A-1)}, \;
 D_2=\frac{  i \bar A}{2(4t^2-n^2)^{1/4}(\bar A-1)}. 
 \label{eq:D1D2}
 \end{align}
Moreover we define,  for $j=1, 2$, 
 \begin{align}
 &
 \chi_j (S_j)=\frac{1}{2\pi i}\int_{\exp(-\pi i/4)}^{S_j} 
                \log \frac{1-|r(\tau)|^2}{1-|r(S_j)|^2}    \frac{\, d \tau}{\tau-S_j}, 
  \\
 &
 \nu_j=-\frac{1}{2\pi} \log (1-|r(S_j)|^2), 
  \\
  &
  \widehat\delta_j (S_j) =
   \exp 
   \left(
     \frac{1}{2\pi}
     \biggl[
         (-1)^{j} \int_{ e^{-\pi  i/4}}^{S_{3-j}}-\int_{-S_1}^{-S_2} 
     \biggr]
     \frac{\log (1-|r(\tau)|^2) }{\tau-S_j}\, d \tau
   \right),
 \\
 & 
 \delta_j^0 
 = S_j^n  e^{{- i t} (S_j-S_j^{-1})^2 /{2}}
  D_j^{  (-1)^{j-1} i \nu_j}                 
         e^{  (-1)^{j-1} \chi_j(S_j)}       \widehat\delta_j (S_j), 
  \label{eq:deltaj0_easy}
\end{align}
where the contours are    minor arcs $\subset C$. 
We have $\mathrm{Re\,} D_j>0$ and $z^{(-1)^{j-1} i \nu_j}$ has 
a cut along the negative real axis. 
See \eqref{eq:defdelta},  \eqref{eq:chi_j} and \eqref{eq:hatdeltajdef} 
for $\delta(z)$,  
$\chi_j(z)$ and $\hat\delta_j(z)$ at a general point $z$ 
(not only for $j=1, 2$ but also for $j=3, 4$).  
Another expression of $\delta_j^0$ is given in \eqref{eq:deltaj0defdecember}. 
We have $\delta(0)\ge 1$ and $\nu_j\ge 0$ since $|r|<1$. 
Notice that  $A, S_j, \delta(0), \chi_j (S_j), \nu_j$ 
and $\hat\delta_j(S_j)$ are functions in $n/t$ and that 
$\beta_j$ and  $D_j$ are of the form 
$t^{-1/2}\times (\mbox{a function in }n/t)$. 
As $t\to\infty$, $\beta_j$ is decaying and $\delta_j^0$  
is oscillatory if $n/t$ is fixed.  

Now we present our main result. Its proof will be given at the 
end of \S\ref{sec:reconstscaling}. 
\begin{theorem}
\label{thm:main}Let $V_0$ be a constant with $0<V_0<2$.  
Assume that  the initial value satisfies 
the smallness condition \eqref{eq:smallnesscondition} 
and the rapid decrease condition \eqref{eq:rapidlydecreasing}. 
Then  in the region 	$|n|\le (2-V_0)t$, 
the asymptotic behavior of the solution to \eqref{eq:IDNLS}    is 
 \begin{equation}
 R_n(t)
 =-\frac{\delta(0)}{\pi i}\sum_{j=1}^2 
 \beta_j (\delta_j^0)^{-2}  S_j^{-2}  M_j
  +O(t^{-1}\log t)
  \hspace{0.5em}\mbox{as}\hspace{0.5em} t\to\infty, 
 \label{eq:main}
 \end{equation}
where 
\[ M_j
  = \frac{ \sqrt{2\pi}  \exp
     \left(
      (-1)^j 3\pi i/4 -\pi \nu_j/2
     \right)
     }
     {
            {\bar r(S_j) \Gamma( (-1)^{j-1}i \nu_j)}
     } 
     \quad \mbox{if } r(S_j)\ne 0
\]
and $M_j=0$ if $r(S_j)=0$. 
The symbol $O$ represents an asymptotic estimate which is 
uniform with respect to $(t, n)$ satisfying $|n|\le (2-V_0)t$.  
Each term in the summation exhibits a behavior of decaying 
oscillation   of order $t^{-1/2}$ as $t\to\infty$ while 
$n/t$ is fixed. 
\end{theorem}

 Notice that $\delta_j^0$ has three oscillatory factors 
$S_j^n,   e^{{- i t} (S_j-S_j^{-1})^2 /{2}}$ and 
$D_j^{  (-1)^{j-1} i \nu_j}$. We claim that $S_j^n$ is oscillatory 
because $n$ tends to infinity together with  $t$  if the ratio $n/t$ is fixed. 
Set $\theta_j=\arg S_j\in (-\pi/2, 0)$, $a_j=\mathrm{Im\,} S_j$. 
Then we have 
\begin{align*}
  &S_j^n  e^{{- i t} (S_j-S_j^{-1})^2 /{2}} 
 D_j^{  (-1)^{j-1} i \nu_j}
 \\
 & =\exp
 \left(
  i n \theta_j+2i a_j^2 t -\frac{(-1)^{j-1}}{2}i \nu_j \log t
 \right)
 \times 
 (\mbox{a function in }n/t). 
\end{align*}
Therefore $\beta_j (\delta_j^0)^{-2}$ in 
\eqref{eq:main} 
behaves like 
\(
 \mathrm{const. }\, t^{-1/2} \exp
 (i (p_j t+q_j \log t))
\)
\(
 (p_j, q_j\in\mathbb{R})
\). 
This is an analogue of the Zakhalov-Manakov formula concerning 
the continuous defocusing NLS (\cite{DZUTYO, Important, Ma, ZMa}). 
Notice that $p_j=p_j(n/t)$ can be either positive or negative 
depending on the ratio $n/t$. This kind of change of sign is not observed 
in the case of the continuous NLS. 

\begin{remark}
A careful inspection of the proof shows that it is enough to 
assume $\sum |n|^s |R_n(0)|<\infty$ for $s=10$. It implies 
that the eigenfunctions and $r$ are in $C^{10}$ on $|z|=1$. 
Therefore $q$ and $k$ in \S\ref{sec:decomposition} 
can be so chosen that $k=9, q=2$ and we can set $\ell=1$ 
in Proposition~\ref{prop:Qnw'2}. 
\end{remark}

\section{A new Riemann-Hilbert problem}

Each $S_j$ is a saddle point of $\varphi$ with $|S_j|=1$ and 
\begin{align*}
 &
 \varphi''(S_1)=\varphi''(S_3)=2\bar{A}^2\sqrt{4t^2-n^2}
    =2 A^{-2}\sqrt{4t^2-n^2},  
  \\
 &
 \varphi''(S_2)=\varphi''(S_4)=-2A^2\sqrt{4t^2-n^2}
    =-2\bar{A}^{-2}\sqrt{4t^2-n^2}
    =-\overline{\varphi''(S_1)}.
\end{align*}

For $z=r e^{ i \theta}$, we have 
\(
 \mathrm{Re\,}\varphi=\frac{-1}{2}t (r^2-r^{-2}) \sin 2\theta - n \log r. 
\) 
It vanishes for any $\theta$ if $r=1$. 
For any other positive value of  $r$, 
the equation $\mathrm{Re\,} \varphi=0$ gives four branches of 
$\theta=\theta(r)\in[0, 2\pi[$ and represents a curve with certain symmetry. It is shown in the Figure \ref{fig:phase}, together with  the sign of  
$\mathrm{Re\,}  \varphi$
in the each region. 

\begin{figure}\centering
 \includegraphics[width=6cm]{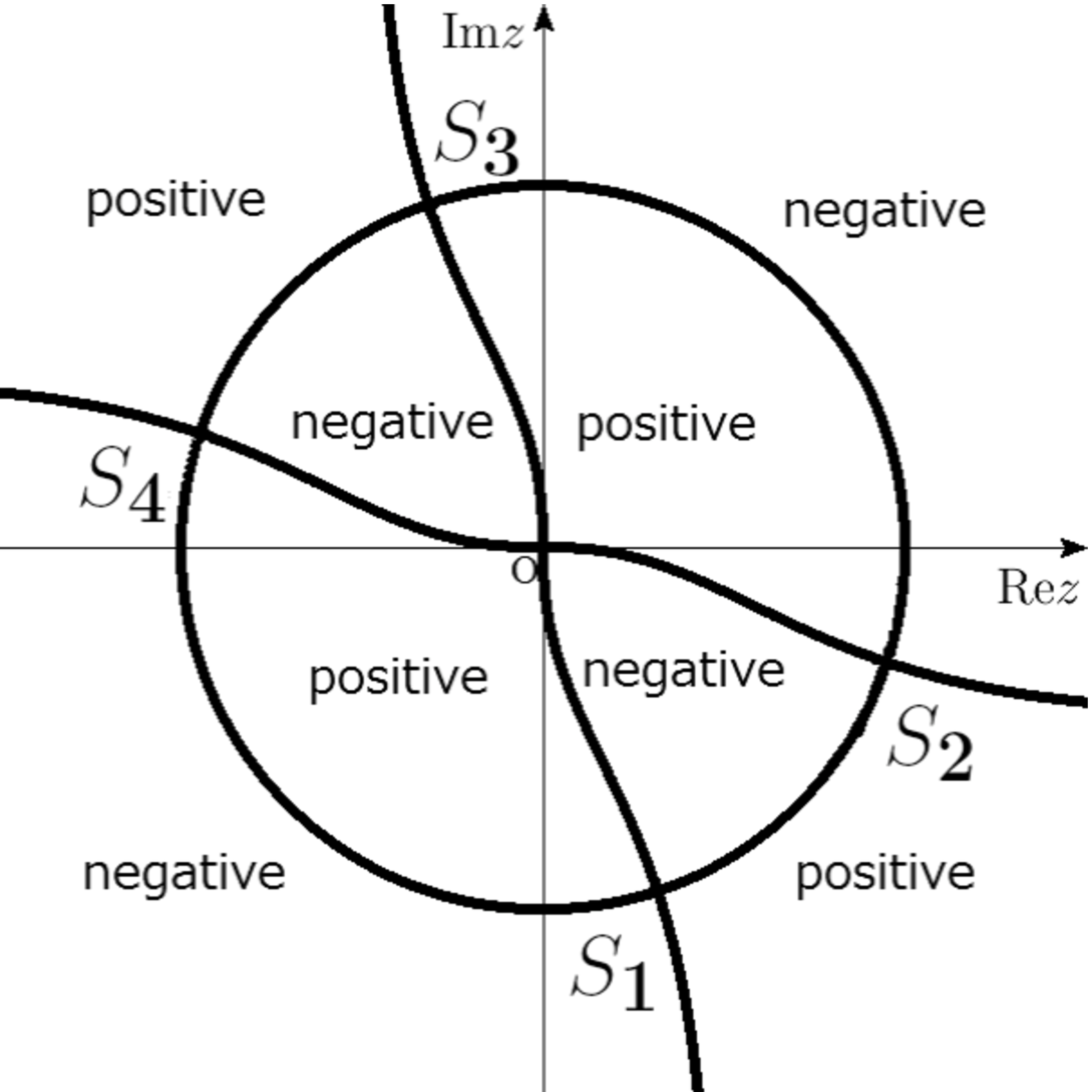}
 \caption{Signs of $\mathrm{Re\,} \varphi$}
 \label{fig:phase}
\end{figure}

Let $\delta(z)$, analytic in $|z|\ne 1$, be the solution to the  Riemann-Hilbert problem   
\begin{align}
 & \delta_+(z)=\delta_-(z) (1-|r(z)|^2) \;  
  \mbox{ on}\, \arc{S_1S_2} \;\mbox{and}\; \arc{S_3S_4}, 
  \label{eq:delta1}
 \\
 &  \delta_+(z)=\delta_-(z)  \; 
  \mbox{ on}\, \arc{S_2S_3} \;\mbox{and}\; \arc{S_4S_1}, 
  \label{eq:delta2}
 \\
 & \delta(z)\to 1  \; 
  \mbox{ as}\; z\to \infty, 
  \label{eq:delta3}
\end{align}
where   $\arc{S_jS_k}$ is the minor  arc  $\subset C$ 
joining $S_j$ and $S_k$ and the \textit{outside} of $C$ is the plus side.
On the singular locus $\arc{S_1S_2}\cup\arc{S_3S_4}$, 
$\delta_\pm(z)$ are the boundary values from 
$\pm\mathrm{Re\,} \varphi>0$ respectively,  and 
there is no distinction between them on  $\arc{S_2S_3}\cup\arc{S_4S_1}$. 

This problem can be uniquely solved by the formula
\begin{equation}
 \delta(z)=
 \exp\left(
  \frac{-1}{2\pi  i}
  \left[\int_{S_1}^{S_2}+\int_{S_3}^{S_4}\right]
   (\tau-z)^{-1} \log (1-|r(\tau)|^2) \, d \tau
  \right), 
 \label{eq:defdelta}
\end{equation}
where the contours are the arcs $\subset C$. 
We have 
$\delta(-z)=\delta(z)$ and $\delta'(0)=0$ because $r(-\tau)=-r(\tau)$.

Conjugating our original Riemann-Hilbert problem   
(\ref{eq:originalRHP1})-(\ref{eq:originalRHP3}) by 
\[ 
 \Delta(z)=
 \begin{bmatrix}
  \delta(z) & \quad  0 \\ 0 & \quad  \delta^{-1}(z)
 \end{bmatrix}
 =
 \delta^{\sigma_3}(z)
\]
leads to the factorization problem for $m\Delta^{-1}$, 
\begin{align}
 &  (m\Delta^{-1})_+ (z) = (m\Delta^{-1})_-(z) 
   	(\Delta_- v \Delta_+^{-1}), 
   	\quad z\in C, 
 \label{eq:DeltaRHP1}
 \\
 & m\Delta^{-1}\to I \quad (z\to\infty).
 \label{eq:DeltaRHP2}
\end{align}

Now, we rewrite 
(\ref{eq:DeltaRHP1})-(\ref{eq:DeltaRHP2}) 
by  choosing  the counterclockwise orientation (the \textit{inside} being the plus side) on   
$\arc{S_2 S_3}$ and $\arc{S_4 S_1}$ and the clockwise orientation  (the outside being the plus side) on $\arc{S_1 S_2}$ and $\arc{S_3 S_4}$. 
The circle $|z|=1$ with this new orientation is denoted by 
$\widetilde C$ and the new Riemann-Hilbert problem on it is
\begin{align}
 &  (m\Delta^{-1})_+ (z) = (m\Delta^{-1})_-(z) 
   	\tilde v, 
   	\quad z\in \widetilde C, 
 \label{eq:mdelta^-1_1}
 \\
 & m\Delta^{-1}\to I \quad (z\to\infty), 
 \label{eq:mdelta^-1_2}
\end{align}
for the 2$\times$2 matrix $\tilde v$. 
We have,  by (\ref{eq:delta1}) and (\ref{eq:delta2}),  
\begin{align*}
  \tilde v
  &= e^{-\varphi \mathrm{ad}\sigma_3}
  \left(
   \begin{bmatrix}
     1&\; 0 \\ r\delta_-^{-2}/(1-|r|^2)\; &\; 1
   \end{bmatrix}
   \begin{bmatrix}
    1\;&\; -\bar{r}\delta_+^2/(1-|r|^2) \\ 0\;& \;1
   \end{bmatrix}
  \right)
  &\mbox{ on } S_1S_2\cup S_3S_4, 
  \\
  &=
   e^{-\varphi \mathrm{ad}\sigma_3}
  \left(
   \begin{bmatrix}
     1\;& \;0 \\ -r\delta^{-2}\;  &\; 1
   \end{bmatrix}
   \begin{bmatrix}
    1\;&\;  \bar{r}\delta^2  \\ 0\;&\; 1
   \end{bmatrix}
  \right)
   &\mbox{ on } S_2S_3\cup S_4S_1.
\end{align*}

Set 
\begin{align}
 \rho & =-\bar{r}/(1-|r|^2) & \mbox{\quad on\;}  S_1S_2\cup S_3S_4, 
 \\
 &  =\bar{r} & \mbox{\quad on\;}  S_2S_3\cup S_4S_1, 
\end{align}
then $\tilde v$ admits a unified expression
\[
   \tilde v=   e^{-\varphi \mathrm{ad}\sigma_3} 
   \left(
   \begin{bmatrix}
     1\;&\; 0 \\ -\bar{\rho}\delta_-^{-2} \; & \; 1
   \end{bmatrix}
   \begin{bmatrix}
    1\;&\;  \rho\delta_+^2  \\ 0\;&\; 1
   \end{bmatrix}
  \right)
\]
on any of the arcs, where $\delta_+=\delta_-=\delta$ 
on $S_2S_3\cup S_4S_1$. 
We have a lower/upper factorization 
\begin{align}
   \tilde v&=  b_-^{-1}  b_+, 
   \label{eq:factorization}
 \\
 b_+&:=\delta_+^{\mathrm{ad}\sigma_3} 
  e^{-\varphi \mathrm{ad}\sigma_3} 
 \begin{bmatrix} 1 \; &\; \rho \\ 0 \;&\; 1\end{bmatrix}
 =\begin{bmatrix} 
        1 \;&\; \delta_+^2  e^{-2\varphi} \rho \\ 0\;&\; 1 
   \end{bmatrix}  , 
 \; 
 \label{eq:upper-triangular}
 \\
 b_-&:=\delta_-^{\mathrm{ad}\sigma_3} 
   e^{-\varphi \mathrm{ad}\sigma_3} 
  \begin{bmatrix} 1 \;&\; 0 \\ \bar{\rho} \;&\; 1 \end{bmatrix} 
 =\begin{bmatrix} 1 \;&\; 0 \\  
           \delta_-^{-2}  e^{2\varphi} \bar\rho \;&\; 1
   \end{bmatrix}    . 
  \label{eq:lower-triangular}
\end{align}
Later we shall use $w_\pm=\pm(b_\pm -I)$.

\section{Decomposition, analytic continuation and estimates}
\label{sec:decomposition}
From now on we assume $0< n \le (2-V_0)t$ so that 
$-\pi/4 < \arg A<0$.  
Minor  modifications are required in the  construction of the 
contour $\Sigma$ (Figure~\ref{fig:contour} below) if 
$-(2-V_0)t \le n \le 0$. See Remark~\ref{rem:negativen}. 

Set 
\[
 \psi=\varphi/( i t)=2^{-1}(z-z^{-1})^2 + i n t^{-1}\log z. 
\]
It is real-valued on $|z|=1$. 
For $z= e^{ i \theta}\,(\theta\in\mathbb{R})$, we have 
$\psi=\cos 2\theta-nt^{-1}\theta-1$ and $\psi$ is monotone 
on any of $S_1S_2$, $S_2S_3$, $S_3S_4$, $S_4S_1$. 
The monotonicity also  follows from the fact 
that  there are no other stationary points of $\varphi$ on $|z|=1$ 
than  $S_j$'s.

\subsection{Decomposition on an arc and some estimates}
\label{subsec:decomposition}

We seek  a decomposition 
$\rho=R+h_I+h_{II}$ with each term having a certain estimate. 
Set $\vartheta=\theta+\pi/4$, 
$\vartheta_0=\arg \bar{A}=\arctan\sqrt{(2t-n)/(2t+n)}$. 
Then $\arc{S_1S_2}$ corresponds to 
$-\vartheta_0\le \vartheta\le \vartheta_0$. 
We regard the function $\rho$ on $\arc{S_1 S_2}$ 
as a function in $\vartheta$ and denote it 
by $\rho(\vartheta)$ by abuse of notation. We have 
\(
 \rho(\vartheta)=H_e(\vartheta^2)+\vartheta H_o(\vartheta^2)\, 
 (-\vartheta_0\le \vartheta\le\vartheta_0)
\) 
for smooth functions $H_e$ and $H_o$. 
By Taylor's theorem, they are expressed as follows:
\begin{align*}
 H_e(\vartheta^2) =\mu_0^e+\cdots+\mu_k^e(\vartheta^2-\vartheta_0^2)^k 
 +\frac{1}{k!}
   \int_{\vartheta_0^2}^{\vartheta^2} H_e^{(k+1)}(\gamma)
   (\vartheta^2-\gamma)^k 
   \, d \gamma, 
 \\ 
  H_o(\vartheta^2) =\mu_0^o+\cdots+\mu_k^o(\vartheta^2-\vartheta_0^2)^k
 +\frac{1}{k!}
   \int_{\vartheta_0^2}^{\vartheta^2} H_o^{(k+1)}(\gamma)
   (\vartheta^2-\gamma)^k
   \, d \gamma. 
\end{align*}
Here $k$ can be any positive integer, but we assume 
$k =4q+1, \, q\in\mathbb{Z}_+$ for convenience of later calculations. 

We set
\begin{align*}
 R(\vartheta)&=R_k(\vartheta)
 =\sum_{i=0}^k \mu_i^e (\vartheta^2-\vartheta_0^2)^i 
  +\vartheta\sum_{i=0}^k \mu_i^o (\vartheta^2-\vartheta_0^2)^i, 
 \\
 \alpha(z)&=(z-S_1)^q (z-S_2)^q, 
 \\
 h(\vartheta)&=\rho(\vartheta)-R(\vartheta)
\end{align*}
and, by abuse of notation, 
\begin{equation*}
  \alpha(\vartheta)
 =\alpha( e^{ i (\vartheta-\pi/4)})
 =[ e^{ i (\vartheta-\pi/4)}- e^{ i (-\vartheta_0-\pi/4)}]^q 
   [ e^{ i (\vartheta-\pi/4)}- e^{ i (\vartheta_0-\pi/4)}]^q .
\end{equation*}
Notice that we have 
\(
 R(\pm \vartheta_0)=\rho(\pm \vartheta_0)
\). 
The function $R$ extends analytically from $\arc{S_1S_2}$ to a 
fairly large complex neighborhood. Its singularity comes only from 
that of $\log z$. 
By abuse of notation, $R(z)$ denotes the analytic function thus 
obtained, so that 
$R(\vartheta)=R(e^{i(\vartheta-\pi/4)})$ and 
$R(S_j)=\rho(S_j)$. 

We have $\, d \psi/\, d \vartheta=2\cos 2\vartheta-n/t$ 
and it has a zero of order 1 at $\vartheta=\pm \vartheta_0$. 
Since $[-\vartheta_0, \vartheta_0]\ni\vartheta\mapsto \psi\in\mathbb{R}$ is 
strictly increasing, we can consider its inverse 
$\vartheta=\vartheta(\psi)$, $\psi(-\vartheta_0)\le\psi\le\psi(\vartheta_0)$.  
We set 
\begin{align*}
 (h/\alpha)(\psi)
 &= h(\vartheta(\psi))/\alpha(\vartheta(\psi)) 
 & \mbox{in}\; 
      \psi(-\vartheta_0)\le \psi \le \psi(\vartheta_0), 
 \\
 &=0 &\mbox{otherwise}.
\end{align*}

Then $(h/\alpha)(\psi)$ is well-defined for $\psi\in\mathbb{R}$, 
and it can be shown that
 $h/\alpha\in H^{(3q+2)/2}(-\infty<\psi<\infty)$ 
and that its norm is uniformly bounded with respect to 
$(n, t)$   with (\ref{eq:region}). 
This argument is a `curved' version of 
\cite[(1.33)]{DZ}. 
Notice that  $\vartheta_0$, the counterpart of $z_0$ of \cite{DZ},  is 
trivially bounded.

Set 
\begin{align*}
 &(\widehat{h/\alpha})(s)
 =
 \int_{-\infty}^{\infty}  e^{- i s \psi} 
 (h/\alpha)(\psi) \frac{\, d \psi}{\sqrt{2\pi}},
 \\
 &h_I (\vartheta)
  =\alpha(\vartheta) \int_t^\infty  e^{ i s \psi(\vartheta)}
  (\widehat{h/\alpha})(s) \frac{\, d s}{\sqrt{2\pi}}, 
 \\
 &h_{II} (\vartheta)
  =\alpha(\vartheta) \int_{-\infty}^t  e^{ i s \psi(\vartheta)}
  (\widehat{h/\alpha})(s) \frac{\, d s}{\sqrt{2\pi}}, 
\end{align*}
then 
$h(\vartheta)=h_I(\vartheta)+h_{II}(\vartheta)$, $|\vartheta|\le\vartheta_0$ and 
\begin{equation}
 | e^{-2 i t \psi}h_I(\vartheta)|\le C/t^{(3q+1)/2}
 \label{eq:h_I}
\end{equation}
for some $C>0$ on  $\arc{S_1S_2}$. See \cite[(1.36)]{DZ}. 
The symbol $C$  always denotes   a generic positive constant.

Set $p=\sqrt{2(2t-n)}/(\sqrt{2t+n\,}-\sqrt{2t-n\,}\,)>0$.  
We consider the contour 
\(
  L_{12}=l_{12}\cup l_{21}\subset
  \{\mathrm{Re\,} \varphi\ge 0\}
  =
  \{\mathrm{Re\,} i\psi\ge 0\}
\), 
where
\begin{align*}
 l_{12}\colon& z(u)=S_1+(p-u)A
   &\; (0\le u \le p), 
 \\
 l_{21}\colon& z(u)=S_2- i u\bar{A}
  &\; (0\le u \le p).
\end{align*}
We have chosen $p$ so that  
$S_1+pA=S_2- i p \bar A$, that is, 
$l_{12}$ and $l_{21}$ are joined at a single point. 

\begin{figure}\centering
\includegraphics[width=8cm]{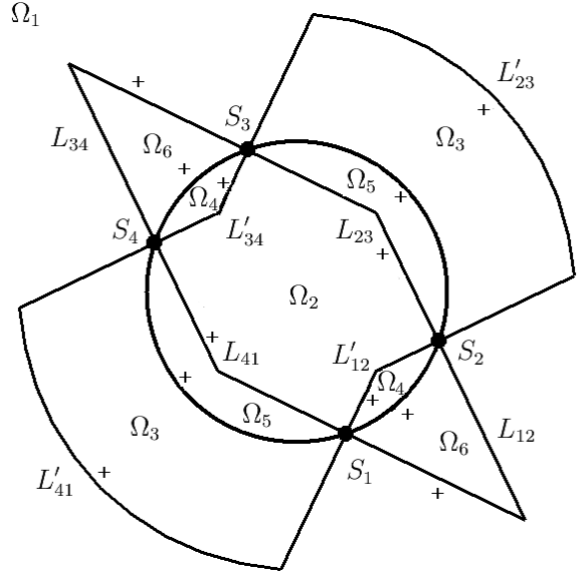}
\caption{the contour $\Sigma$}
\label{fig:contour}
\end{figure}

We can show that $h_{II}(\vartheta)$ can be analytically continued 
to $\{\mathrm{Re\,} \varphi>0\}=\{\mathrm{Re\,} i\psi> 0\}$.  
The extension is denoted by $h_{II}(z)$ 
so that $h_{II}( e^{ i(\vartheta-\pi/4)})=h_{II}(\vartheta)$ 
by abuse of notation. 
On $l_{12}$, we have for $v=p-u(=$the distance from $S_1)$, 
\begin{align*}
  {}&|\alpha(z(u))|
 \le v^q |S_1-S_2+vA|^q \le C^q v^q, 
 \\
 {} &| e^{-2 i t \psi(z(u))} h_{II}(z(u))|
 \le \frac{C^q v^q  e^{-t\mathrm{Re\,}  i \psi}}{\sqrt{2\pi}}
        \left(
        	\int_{-\infty}^t \frac{\, d s}{1+s^2}
        \right)^{1/2}
 \\
 &\hspace{10em} \times
        \left(
        	\int_{-\infty}^t
        	 (1+s^2) 
        	|(\widehat{h/\alpha})(s)|^2 \, d s  
        \right)^{1/2}.
\end{align*}
We have 
\begin{equation}
 \mathrm{Re\,}  i \psi\ge C'v^2 
\label{eq:estimatefrombelowbysquare}
\end{equation}
for some $C'>0$, because $v=0$ corresponds to the saddle point $S_1$. 
On $l_{12}$, 
\begin{equation*}
 | e^{-2 i \psi(z(u))} h_{II}(z(u))|
 \le C  v^q  e^{-C'tv^2}
 =Ct^{-q/2}(tv^2)^{q/2}  e^{-C' tv^2}
 \le C/t^{q/2}.
\end{equation*}
We have a similar estimate on $l_{21}$. Therefore, all over $L_{12}$, we have
\begin{equation}
 | e^{-2 i t \psi(z)}h_{II}(z)|
 \le C/t^{q/2}.
\label{eq:h_II}
\end{equation}

For a small  constant $\varepsilon>0$, 
let $l_{12}^\varepsilon\subset l_{12}$ and 
$l_{21}^\varepsilon\subset l_{21}$ be the segments given by
\begin{align*}
 & l_{12}^\varepsilon \colon 
    z(u)=S_1+(p-u)A
 &\quad (0\le u \le p-\varepsilon),
 \\
 & l_{21}^\varepsilon \colon 
    z(u)=S_2- i u \bar{A}
 &\quad (0<\varepsilon\le u \le p).
\end{align*}
The segment $l^\varepsilon_{jk}$ is obtained by removing the 
$\varepsilon$-neighborhood of $S_j$ from $l_{jk}$. 
On $l_{12}^\varepsilon  \cup l_{21}^\varepsilon\subset L_{12}$, 
we have 
$\mathrm{Re\,}  i \psi \ge C_\varepsilon \varepsilon^2$ for 
some $C_\varepsilon>0$. It implies 
\begin{equation}
	 | e^{-2 i t \psi(z)} R(z)|\le C 
	  e^{-C_\varepsilon \varepsilon^2 t}.
\label{eq:R}	 
\end{equation}

\subsection{Decomposition of another function on the same arc}

The function $\bar\rho$ on $\arc{S_1S_2}$ can be decomposed 
as $\bar\rho=\bar R+\bar h_I+\bar h_{II}$.

Set $p'=\sqrt{2(2t-n)}/(\sqrt{2t+n\,}+\sqrt{2t-n\,}\,)>0$. 
We construct the contour 
$L'_{12}=l'_{12}\cup l'_{21}$ in 
$\{\mathrm{Re\,} \varphi\le0\}=\{\mathrm{Re\,} i\psi\le0\}$ as follows: 
\begin{align*}
 l'_{12}\colon& z(u)=S_1+ (p'-u) i   A
   &\; (0\le u \le p'), 
 \\
 l'_{21}\colon& z(u)=S_2- u \bar{A}
  &\; (0\le u \le p').
\end{align*}
We have chosen $p'$ so that 
$l'_{12}$ and $l'_{21}$ are joined at a single point 
$S_1+ i p' A=S_2-p' \bar A$.

In the same way as (\ref{eq:h_I}) and  (\ref{eq:h_II}), we can show 
\begin{align}
 &| e^{2 i t \psi}\bar h_I(\vartheta)|\le C/t^{(3q+1)/2}
 \; \mbox{ on}\; \arc{S_1S_2}, 
 \label{eq:barh_I}
 \\
 \label{eq:barh_II}
 & | e^{2 i t \psi(z)}\bar h_{II}(z)|
 \le C/t^{q/2} \; \mbox{ on}\; L'_{12}.
\end{align}

Set 
\begin{align*}
 l_{12}^{\prime\varepsilon}\colon& z(u)=S_1+ (p'-u) i  A
   &\; (0\le u \le p'-\varepsilon), 
 \\
 l_{21}^{\prime\varepsilon}\colon& z(u)=S_2- u \bar{A}
  &\; (\varepsilon\le u \le p').
\end{align*}
In the same way as (\ref{eq:R}), we can show that 
\begin{equation}
	 | e^{2 i t \psi(z)} \overline{R}(z)|\le C 
	  e^{-C_\varepsilon \varepsilon^2 t}
	 \;\mbox{on}\;
	  l_{12}^{\prime\varepsilon}\cup l_{21}^{\prime\varepsilon}.
\label{eq:barR}
\end{equation}

\subsection{Decomposition on another arc}
\label{subsec:anotherarc}

On $\arc{S_2S_3}$, the functions 
$R, h_I, h_{II}, \bar R, \bar h_I, \bar h_{II}$ 
are constructed from $\rho$ and $\bar\rho$ 
in the same way as above. We have 
\begin{equation*}
  | e^{-2 i t \psi}h_I|\le C/t^{(3q+1)/2}, 
  \quad
  | e^{2 i t \psi}\bar h_I|\le C/t^{(3q+1)/2}.
\end{equation*}

Set $p''=\sqrt{2(2t+n)}/(\sqrt{2t+n\,}+\sqrt{2t-n\,}\,)>0$. 
Let $L_{23}$ be the contour obtained by joining 
\begin{align*}
 l_{23}\colon &z(u)=S_2+ i u \bar{A}
 &\quad (0\le u \le p''),
 \\
 l_{32}\colon &z(u)=S_3+(p''-u)A
 &\quad (0\le u \le p'').
\end{align*}
The segments $l_{jk}^\varepsilon\subset l_{jk},\, (j, k)=(2, 3), (3, 2)$ consist 
of points whose distance from $S_j$ is not less than $\varepsilon$.  
Notice that $S_2+ i p''\bar A=S_3+p'' A$ is inside the circle $|z|=1$. 
Then we can show in the same way as (\ref{eq:h_II})  and \eqref{eq:R} that 
\begin{align*}
  &| e^{-2 i t \psi(z)}h_{II}(z)|\le C/t^{q/2} \quad\mbox{on } L_{23}.
  \\
	& | e^{-2 i t \psi(z)} R(z)|\le C 
	  e^{-C_\varepsilon \varepsilon^2 t}
    \quad\mbox{on }
   l_{23}^\varepsilon\cup l_{32}^\varepsilon.
\end{align*}

Next,  let $L'_{23}$ be the contour obtained by joining
\begin{align*}
  l'_{23}\colon& z(u)=S_2+   u \bar A & (0\le u\le 1), \\
  \hat l'_{23}\colon& z(u)=
         \mathrm{Arc}(S_2+\bar A, S_3+iA), 
\;
        & \\ 
  l'_{32}\colon& z(u)=S_3 +  i (1-u) A & (0\le u\le 1). 
\end{align*}
Here $\mathrm{Arc}(S_2+\bar A, S_3+iA)$ is the minor arc 
$\subset \{|z|=(2+\sqrt{2})^{1/2}\}$ from $S_2+\bar A$ 
to $ S_3+iA$. 
Then we have estimates on $l'_{23}\cup l'_{32}$ similar to 
(\ref{eq:barh_II}). 
The arc $\hat l'_{23}$   is away from 
the  saddle points and we have  
exponential decay of $ e^{2 i t \psi}\bar h_{II}$ on it. 
Therefore we get an estimate like 
 (\ref{eq:barh_II}) on $L'_{23}$. Moreover, we  obtain an estimate like 
(\ref{eq:barR}) if we exclude 
the $\varepsilon$-neighborhood 
of $S_2$ and $S_3$. 

\subsection{Decomposition on the remaining arcs}
We construct $L_{jk}$ and $L'_{jk}$ for 
$(j, k)=(3, 4), (4, 1)$ by symmetry  and get relevant estimates. 
The results in this section lead to   Lemma~\ref{lem:wawbwcw'estimates} below.

\section{A Riemann-Hilbert problem on a new contour}
\label{sec:newcontour}

Set $L=\cup_{(j, k)}L_{jk}\subset\{\mathrm{Re\,}  \varphi\ge 0\}$, 
 $L'=\cup_{(j, k)}L'_{jk}\subset\{\mathrm{Re\,}  \varphi\le 0\}$, 
 $\Sigma=\tilde C \cup L \cup L'$.  
We define six open sets $\Omega_1, \dots, \Omega_6$ as in 
Figure~\ref{fig:contour}. 
The + signs  indicate  the plus  sides of the curves. 
If $j$ is odd (resp. even), 
$\arc{S_j S_{j+1}}$ is oriented clockwise (resp. counterclockwise)  and 
$L\cup L'$ is oriented inward (resp. outward) near $S_j$.

\begin{remark}
\label{rem:negativen}
If $n\le 0$, the contour $\Sigma$ should be slightly modified. 
\begin{enumerate}
\item[(i)] If $n=0$, we should replace $L_{12}$ and $L_{23}$ with curves like 
$L'_{23}$ and $L'_{41}$ in Figure~\ref{fig:contour}, 
each made up of two segments and an arc. 
\item[(ii)]
If $n$ is negative, the argument of $A$ is between 
$-\pi/4$ and $-\pi/2$. The parts $L'_{23}$ and $L'_{41}$ 
should each be made up of two segments, while 
$L_{12}$ and $L_{34}$ of two segments and an arc. 
\end{enumerate}

The expansion in Theorem~\ref{thm:main} is uniform with respect to 
the ratio $n/t\in [-(2-V_0), 2-V_0]$. 
The uniformity can be proved by  considering 
the following three cases   ($\epsilon>0$ is small): (a)  $\epsilon\le n/t\le 2-V_0$, 
(b)  $-\epsilon< n/t< \epsilon$, (c) $-(2-V_0) \le n/t\le  -\epsilon$. 
The uniformity in  (a) follows from the calculations explicitly 
given in the present paper. 
The contour of the type in (i) [resp. (ii)]  is useful in (b) [resp. (c)]. 
In dealing with each part of the contour, 
one has only to follow the subsection~\ref{subsec:decomposition} 
(two segments) 
or \ref{subsec:anotherarc} (two segments and an arc). 

One can show the uniformity in $-(2-V_0)\le n/t \le 2-V_0$ 
by using only one multipurpose contour of the type in  (i). 
See Figure~\ref{fig:Sigmamultipurpose}. 
We have chosen to use 
the contour as in Figure~\ref{fig:contour} in order to simplify 
the presentation in the case (a). 
\begin{figure}\centering
  \includegraphics[width=3cm]{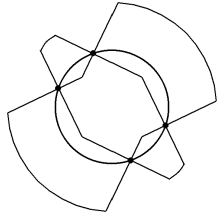} 
  \caption{$\Sigma$ (multipurpose version) }
  \label{fig:Sigmamultipurpose}
\end{figure}

\end{remark}

Notice that each of $\Omega_3, \dots, \Omega_6$ 
has two connected components and that $\Omega_1$ is unbounded. 
We introduce the following matrices:
\begin{align*}
 & 
	 b_+^0=
	 \delta_+^{\mathrm{ad}\sigma_3}  
       \rm e^{-\varphi\mathrm{ad}\sigma_3} 
	 \begin{bmatrix} 1 & \quad  h_I \\ 0 & \quad 1 \end{bmatrix},  \;
	 &
	 b_+^a=
	 \delta_+^{\mathrm{ad}\sigma_3}  
       \rm e^{-\varphi\mathrm{ad}\sigma_3} 
	 \begin{bmatrix} 1\; &\; h_{II}+R \\ 0\; &\; 1 \end{bmatrix}, 
 \\ 
 & 
	 b_-^0=
	 \delta_-^{\mathrm{ad}\sigma_3}  
       \rm e^{-\varphi\mathrm{ad}\sigma_3} 
	 \begin{bmatrix} 1 & \quad 0 \\ \bar{h}_I & \quad 1 \end{bmatrix},  \;
	 &
	 b_-^a=
	 \delta_-^{\mathrm{ad}\sigma_3}  
       \rm e^{-\varphi\mathrm{ad}\sigma_3} 
	 \begin{bmatrix} 1 \;& \; 0 \\ \bar h_{II}+\bar R \; & \; 1 \end{bmatrix}. 
\end{align*}
Then  $b_\pm=b^0_\pm b^a_\pm$, 
$\tilde v=(b_-^a)^{-1} (b_-^0)^{-1} b_+^0 b_+^a$. 
Define a new unknown matrix  $m^\sharp$ by 
\begin{align}
m^\sharp
&= m\Delta^{-1},  
&  z\in\Omega_1\cup\Omega_2, 
\label{eq:msharpdef1}
\\
&= m\Delta^{-1} (b_-^a)^{-1}, 
& z\in\Omega_3\cup\Omega_4, 
\label{eq:msharpdef2}
\\
&= m\Delta^{-1} (b_+^a)^{-1},  
&  z\in\Omega_5\cup\Omega_6 .
\label{eq:msharpdef3}
\end{align}
By (\ref{eq:mdelta^-1_1}),  (\ref{eq:mdelta^-1_2}) and 
(\ref{eq:factorization}) it  is the unique solution to the  
Riemann-Hilbert problem 
\begin{align}
 &  m^\sharp_+ (z)=m^\sharp (z)_- v^\sharp (z), 
   & z\in \Sigma, 
 \label{eq:msharp1}
 \\
 & m^\sharp(z)\to I   &\mbox{as\;} z\to\infty.
  \label{eq:msharp2}
\end{align}
Here $v^\sharp=v^\sharp(z)=v^\sharp(z; n, t)$ is defined by 
\begin{align*}
 v^\sharp(z)
 &= (b_-^0)^{-1} b_+^0, 
 & z\in\tilde C, 
 \\
 &=b_+^a, 
 &  z\in L, 
 \\
 &= (b_-^a)^{-1}, 
 &\qquad z\in L'. 
\end{align*}

Set 
\( b^\sharp_-
   = b_-^0, I, b_-^a
\)   
and 
\( b^\sharp_+
  = b_+^0,  b_+^a, I
\) 
on $\tilde C, L, L'$ respectively. 
Then, on $\Sigma$, we have
\begin{equation*}
 v^\sharp=v^\sharp(z)=(b_-^\sharp)^{-1}    b_+^\sharp.
\end{equation*}
In the next section, 
we shall employ  
\(
 w_\pm^\sharp
  = \pm (b_\pm^\sharp -I )
\), 
\(
  w^\sharp
  =w_+^\sharp+w_-^\sharp
\).  
We have 
\(
 v^\sharp=(I-w^\sharp_-)^{-1} (I+w^\sharp_+)
          =(I+w^\sharp_-) (I+w^\sharp_+)
\).

\section{Reconstruction and a resolvent}

\subsection{Reconstruction}

Since $\delta'(0)=0$, 
(\ref{eq:Rnreconstruction}) and (\ref{eq:msharpdef1}) imply 
\begin{equation}
 R_n(t)=-\lim_{z\to 0}\frac{ 1}{z}(m^\sharp_{21}\delta)
 =-\frac{ \, d}{\, d z} (m^\sharp_{21}\delta)\bigr|_{z=0}
 =-\delta(0) 
  \frac{\, d m^\sharp_{21}}{\, d z} (0). 
\label{eq:Rnreconst}
\end{equation}

Let 
\[
 (C_\pm f)(z)
 =\int_\Sigma \frac{ f(\zeta)}{\zeta-z_\pm}\frac{ \, d \zeta}{2\pi  i}, \;
 =
 \lim_{  \genfrac{}{}{0pt}{}{z'\to z}{z'\in\{\pm\textrm{\scriptsize -side of }\Sigma\}}}
 \int_\Sigma \frac{ f(\zeta)}{\zeta-z'}\frac{ \, d \zeta}{2\pi  i}, \;
   z\in\Sigma
\]
be the Cauchy operators on $\Sigma$. Define 
\(
 C_{w^\sharp}=C_{w^\sharp}^\Sigma\colon L^2(\Sigma)\to L^2(\Sigma)
\) 
by 
\begin{equation}
 C_{w^\sharp} f =C_{w^\sharp}^\Sigma f
 = C_+(f w^\sharp_-) + C_-(f w^\sharp_+)
 \label{eq:integopdef}
\end{equation}
for a $2\times 2$ matrix-valued function $f$. 
Later we will define similar operators by replacing the pair 
$(\Sigma, w^\sharp)$ with others. Even if a kernel, say $\omega_\pm$,  is 
supported by a subcontour $\Sigma_1$ of 
$\Sigma_2$, it is necessary to distinguish between 
$C_\omega^{\Sigma_1}$ and $C_\omega^{\Sigma_2}$. 

Let $\mu^\sharp$ be the solution to the equation
\begin{equation}
 \mu^\sharp=I+ C_{w^\sharp}\mu^\sharp.
\end{equation}
Then we have $\mu^\sharp=(1-C_{w^\sharp})^{-1}I$ 
(the resolvent exists),  and 
\begin{equation}
 m^\sharp(z; n, t)=I+\int_\Sigma 
                    \frac{\mu^\sharp (\zeta; n, t) w^\sharp(\zeta) }{\zeta-z}
                  \frac{ \, d \zeta}{2\pi  i}, \;
 z\in \mathbb{C}\setminus\Sigma
\label{eq:msharpdef}
\end{equation}
is the unique solution to the Riemann-Hilbert problem (\ref{eq:msharp1}), (\ref{eq:msharp2}). 
By substituting (\ref{eq:msharpdef}) into (\ref{eq:Rnreconst}), we find that 
\begin{align}
  \qquad R_n(t)
 &=-\delta(0)\int_\Sigma 
   \zeta^{-2}
  \left[
   \mu^\sharp (\zeta; n, t) w^\sharp(\zeta) 
  \right]_{21}
  \frac{ \, d \zeta}{2\pi  i}
  \\
 {} &
  =
  - \delta(0)
  \int_\Sigma 
   z^{-2}
  \left[
   \bigl((1-C_{w^\sharp})^{-1}I\bigr)(z)
    w^\sharp(z) 
  \right]_{21}
  \frac{ \, d z}{2\pi  i}.
 \label{eq:20111125}
\end{align}
In \S\ref{sec:crosses}, we will prove that the resolvent 
$(1-C_{w^\sharp})^{-1}\colon L^2(\Sigma)\to L^2(\Sigma)$ indeed 
exists for any sufficiently large $t$ and that  its norm is 
uniformly bounded.

\subsection{Partition of the matrices}

We set,  for $0<\varepsilon<\inf_{(n, t)} p'$, 
\begin{align*}
 \textstyle
 L_\varepsilon
 &=\{z\in  L  ;   |z-S_j|>\varepsilon \; \mbox{for all}\; j=1, 2, 3, 4\}, 
 \\
 L'_\varepsilon
 &=\{z\in L' ;  |z-S_j|>\varepsilon  \; \mbox{for all}\; j=1, 2, 3, 4\}.
\end{align*}
Set $w^a_\pm=w^\sharp_\pm|_{\tilde C}$  
(consisting of quantities of type $h_I$ or $\bar h_I$ together with 
$\delta$ and $ e^{\pm\varphi}= e^{\pm i t \psi}$)   
on $\tilde C$ and 
$w^a_\pm=0$ on $\Sigma\setminus \tilde C$. 
Let $w^b_\pm$ be equal to the contribution to 
$w^\sharp_\pm$ from the 
quantities involving $h_{II}$ or $\bar h_{II}$ on $L\cup L'$ 
and set $w^b_\pm=0$ on $\Sigma\setminus (L\cup L')$. 
Additionally, 
let $w^c_\pm$ be equal to the contribution to 
$w^\sharp_\pm$ from the quantities involving $R$ or $\bar R$ on 
$L_\varepsilon \cup L'_\varepsilon$ 
and set $w^c_\pm=0$ on 
$\Sigma\setminus (L_\varepsilon\cup L_\varepsilon')$. 
Finally, we set 
\(
  w^e_\pm  =w^a_\pm  +w^b_\pm+w^c_\pm
\) 
and 
\(  w'_\pm    = w^\sharp_\pm - w^e_\pm 
\). 
These matrices are all upper or lower triangular and their 
diagonal elements are zero. 
We will show that $w_\pm^e$ are small in a certain sense and that 
the main contribution is by $w'_\pm$.

We define 
\(
  w^*=w^*_+ +w^*_-
\) 
for $*=a, b, c, e, '$.
Then we have 
\(
  w^e
  =w^a+w^b+w^c
\) 
and 
$w'=w^\sharp-w^e$. 
Set 
\(
 \Sigma' 
 =\{
      \Sigma\setminus(\tilde C \cup L_\varepsilon \cup L'_\varepsilon)
    \}
   \cup     \{S_1, S_2, S_3, S_4\}
 =(L\cup L')\cap
     \cup_{j=1}^4 \{|z-S_j|\le\varepsilon\}
\), 
then it is a union of four small crosses  and 
\(
     \mathrm{supp\,}w'_\pm \subset
     \mathrm{supp\,}w'\subset
     \Sigma'
\). As is shown in Figure~\ref{fig:Sigma'}, each cross is 
oriented inward or outward. 
We have
\begin{align}
  w'_+
 &= \delta^{\mathrm{ad\,}\sigma_3}
  e^{-\varphi{\mathrm{ad\,}\sigma_3}}
 \begin{bmatrix} 0 & \; R \\ 0 & \; 0 \end{bmatrix}
 = \begin{bmatrix} 0 \;& \quad \delta^2  e^{-2 i t \psi} R \\ 0 \;& 0 \end{bmatrix}
 &\quad\mbox{on\,}L\cap\Sigma', 
 \label{eq:w'+}
 \\
  w'_-
 &= \delta^{\mathrm{ad\,}\sigma_3}
  e^{-\varphi{\mathrm{ad\,}\sigma_3}}
 \begin{bmatrix} 0 &\;0 \\ - \bar R & \; 0 \end{bmatrix}
 = \begin{bmatrix} 0 & \quad 0 \\ - \delta^{-2}  e^{2 i t \psi}\bar R & \quad 0 \end{bmatrix}
 &\quad\mbox{on\,}L'\cap\Sigma'
  \label{eq:w'-}
\end{align}
and they vanish elsewhere.

\begin{figure}\centering
  \includegraphics[width=5cm]{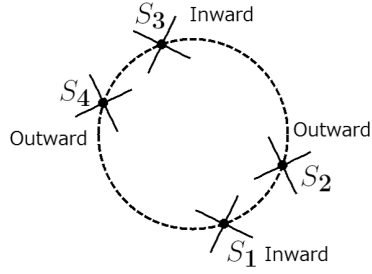}
 \caption{$\Sigma'$: four small crosses}
 \label{fig:Sigma'}
\end{figure}

\begin{lemma}
 For any positive integer $\ell$, there exist  positive constants 
$C$ and $\gamma_\varepsilon$ such that 
 \begin{align}
   |w^a_\pm|, \, 
   |w^a|\le Ct^{-\ell}
   &\;\mbox{on}\;  
   \tilde C, 
   \label{eq:w^a}
   \\
   |w^b_\pm|, \, 
  |w^b|\le Ct^{-\ell}
   & \;\mbox{on}\; 
   L\cup L',
   \label{eq:w^b}
   \\
   |w^c_\pm|, \, 
  |w^c|\le C  e^{-\gamma_\epsilon t}
   & \;\mbox{on}\; 
  L_\varepsilon \cup   L'_\varepsilon . 
  \label{eq:w^c}
 \end{align}
$L^p$ estimates are easily obtained 
since the length of $\Sigma$ 
is bounded uniformly with respect to $(n, t)$ satisfying 
(\ref{eq:region}).  
Moreover we have
\begin{align}
  & \|w'_\pm \|_{L^2(\Sigma)} \le Ct^{-1/4}, \; 
  \|w'\|_{L^2(\Sigma)} \le Ct^{-1/4},
   \label{eq:w'L2L2estimates}
  \\
  & \|w'_\pm\|_{L^1(\Sigma)} \le Ct^{-1/2}, \; 
  \|w'\|_{L^1(\Sigma)} \le Ct^{-1/2}.
  \label{eq:w'L2L2estimates2}
\end{align}
\label{lem:wawbwcw'estimates}
\end{lemma}

\begin{proof} 
 The boundedness of $\delta$ and $\delta^{-1}$ will be proved 
 in \S\ref{sec:saddlepoints}. 
 The inequality (\ref{eq:w^a}) follows from (\ref{eq:h_I}), 
 (\ref{eq:barh_I}) and their analogues. 
 The   inequalities \eqref{eq:w^b} and \eqref{eq:w^c} are 
 consequences of (\ref{eq:h_II}),  (\ref{eq:R}),  (\ref{eq:barh_II}), 
 (\ref{eq:barR}) and their analogues. 
Finally in order to derive (\ref{eq:w'L2L2estimates}) and  (\ref{eq:w'L2L2estimates2}), 
we employ (\ref{eq:estimatefrombelowbysquare}) and its analogues. 
Since $R, \bar R, \delta, \delta^{-1}$ are bounded, we have only to 
calculate the Gauss type integral 
$\int_{0}^\infty  e^{-{\textrm{\scriptsize const.}}t v^2 } \, d v$. 
\end{proof}

We define the integral operators 
 $C_{w^\prime}$ and   $C_{w^e}$  from $L^2(\Sigma)$ to itself 
 of the type (\ref{eq:integopdef}). 
 We have 
 \(
   C_{w^\sharp}=
   C_{w^\prime}+C_{w^e}
 \). 
 
Later in \S\ref{sec:crosses} we will prove that 
$(1-C_{w^\prime})^{-1}$ and $(1-C_{w^\sharp})^{-1}$  
exist and are uniformly bounded for any sufficiently large $t$. 
We proceed assuming this assertion.

\subsection{Resolvents and estimates}
By the second resolvent identity, or rather 
by (\ref{eq:resolventidentity}) below, we get 
\begin{align}
& \int_\Sigma  z^{-2} \bigl((1-C_{w^\sharp})^{-1}I\bigr)(z)
 w^\sharp(z) 
 \nonumber
 \\
 & =\int_\Sigma  z^{-2} \bigl((1-C_{w^\prime})^{-1}I\bigr)(z)
 w^\prime(z) 
 +\mathrm{I}+\mathrm{II}+\mathrm{III}+\mathrm{IV}, 
 \label{eq:I_II_III_IV}
\end{align}
where
\begin{align*}
  &\mathrm{I}=\int_\Sigma z^{-2}  w^e, \quad 
 \mathrm{II}=\int_\Sigma z^{-2} 
    (1-C_{w^\prime})^{-1}
    (C_{w^e}I) w^\sharp, 
  \\
 & \mathrm{III}=\int_\Sigma z^{-2} 
    (1-C_{w^\prime})^{-1}
    (C_{w^\prime}I) w^e,
 \\
 & \mathrm{IV}=\int_\Sigma z^{-2} 
    (1-C_{w^\prime})^{-1} C_{w^e}
     (1-C_{w^\sharp})^{-1}
    (C_{w^\sharp}I) w^\sharp.
\end{align*}
Since $|z^{-2}|$ is uniformly bounded on $\Sigma$,  (\ref{eq:w^a}),  (\ref{eq:w^b}) and  (\ref{eq:w^c}) imply 
\begin{equation}
 |\mathrm{I}|, |\mathrm{II}|, |\mathrm{III}|, |\mathrm{IV}|\le Ct^{-\ell} \quad\mbox{as}\;t\to\infty. 
  \label{eq:20111125_2}
\end{equation}
\begin{proposition}
 \label{prop:Qnw'}
 We have
 \begin{equation}
  R_n(t)
  =  -\delta(0)
  \int_\Sigma 
  z^{-2}
  \left[
   \bigl((1-C_{w^\prime})^{-1}I\bigr)(z) w^\prime(z) 
  \right]_{21}
  \frac{ \, d z}{2\pi  i}+O(t^{-\ell}).
\end{equation}
\end{proposition}
\begin{proof} Use (\ref{eq:20111125}), 
\eqref{eq:I_II_III_IV} and  (\ref{eq:20111125_2}).  \end{proof}

\begin{proposition}
 \label{prop:Qnw'2}
 We have
 \begin{equation}
  R_n(t)
  =  -\delta(0)
  \int_{\Sigma'} 
  z^{-2}
  \left[
   \bigl((1-C_{w'}^{\Sigma'})^{-1}I\bigr)(z) w^\prime(z) 
  \right]_{21}
  \frac{ \, d z}{2\pi  i}+O(t^{-\ell}). 
\end{equation}
Here $C_{w'}^{\Sigma'}$ is an integral operator on $L^2(\Sigma')$ whose kernel 
is  $w'_\pm|_{\Sigma'}$. Notice that $C_{w'}$  is an integral operator 
on $L^2(\Sigma)$ whose kernel is  $w'_\pm$, matrices of functions on $\Sigma\supset\Sigma'$.  
\end{proposition}
\begin{proof}
Since $\mathrm{supp}\, w'_\pm\subset\Sigma'$, 
we can replace $\int_\Sigma$ by $\int_{\Sigma'}$. Apply  \cite[(2.61)]{DZ}. 
\end{proof}

\begin{remark} 
In deriving  \eqref{eq:I_II_III_IV}, we used the following formula: 
if the operators $A, B, C$ and the matrices $f, g, h$ are 
such that $A=B+C, f=g+h$, then 
\begin{align}
 {} \{(1-A)^{-1}I\}f
 &=\{(1-B)^{-1}I\}g+h+\{(1-B)^{-1}(CI)\}f
 \nonumber
 \\
{} &\quad
 +\{(1-B)^{-1}(BI)\}h + \{(1-B)^{-1}C(1-A)^{-1}(AI)\}f. 
\label{eq:resolventidentity}
\end{align}
In  \eqref{eq:I_II_III_IV}, $f, g, h$ involve the factor $z^{-2}$, 
which is absent in the calculation of \cite{DZ}. 
Its presence complicates matters, and it will be dealt with in  \eqref{eq:difficultineq}. 
\end{remark}

\section{Saddle points and scaling operators}
\label{sec:saddlepoints}

\subsection{Some functions characterizing  arcs and their boundedness}

Set 
$T_1=T_2= e^{-\pi i/4}$,  
$T_3=T_4= e^{3\pi i/4}$ and  
\begin{align}
 & \chi_j(z)
 =\frac{1}{2\pi  i}\int_{T_j}^{S_j} 
   \log 
   \frac{1-|r(\tau)|^2}{1-|r(S_j)|^2}
   \frac{\, d \tau}{\tau-z}, 
 \label{eq:chi_j}
 \\
  & \ell_j(z)
   =\int_{T_j}^{S_j}  \frac{\, d \tau}{\tau-z}, 
  \\
 &\nu_j
 =-\frac{1}{2\pi}\log(1-|r(S_j)|^2)\ge 0 
\end{align}
for $j=1, 2, 3, 4$. 
The integral  $\int_{T_j}^{S_j}$  is  performed along the minor arc 
from $T_j$ to $S_j$, which we denote by $\arc{T_j S_j}$. 
The integral  $\chi_j(S_j)$ is well-defined because the logarithm 
vanishes at $S_j$. 
Moreover, $\chi_j$ and $\ell_j$ are analytic  in the complement 
of    $\arc{T_j S_j}$,  in particular  near $S_k\,(k\ne j)$. 
We have $\nu_1=\nu_3, \nu_2=\nu_4$ and 
\begin{equation}
 \frac{1}{2\pi i}\int_{T_j}^{S_j} 
 \log(1-|r(\tau)|^2) \frac{\, d \tau}{\tau-z}
 = i \nu_j \ell_j(z)+\chi_j(z).
\end{equation}
Set 
\begin{align}
  \delta_j(z)
 & 
  =\exp
  \Bigl(
    (-1)^{j-1}  (i \nu_j \ell_j(z)+\chi_j(z))
  \Bigr)	   
    \\
 & 
 =
  \left(
    \frac{z-S_j}{z-T_j}
  \right)^{  (-1)^{j-1}  i \nu_j}
   e^{  (-1)^{j-1} \chi_j(z)}, 
  \label{eq:delta_j}
\end{align}
where $w^{ i \nu_j}$ is cut along 
$\mathbb{R}_-$ and is positive  on $\mathbb{R}_+$. 
It is analytic in the complement of    $\arc{T_j S_j}$ 
and satisfies a Riemann-Hilbert problem similar to 
\eqref{eq:delta1}-\eqref{eq:delta3}.  
The function $\delta(z)$ in \eqref{eq:defdelta} is decomposed as 
\[
 \delta(z)=
 \prod_{j=1}^4 \delta_j(z).
\]

Since $\mathrm{Im\,} \ell_j(z)=\arg [(z-S_j)/(z-T_j)]$, 
we see that 
$\exp(\pm i \nu_j \ell_j(z))$ is bounded. 
Let $V_j, U_j \subset \arc{T_j S_j}$ be sufficiently small 
neighborhoods of $T_j$ and $S_j$ respectively. 
Then $\chi_j(z)$ and its boundary values on 
$\arc{T_j S_j}\setminus\{U_j, V_j\}$  
are bounded as is proved by the Plemelj formula (\cite{AbCl, AbFokas}). 
This formula involves a principal value integral. 
Its boundedness in $U_j\setminus\{S_j\}$ (as $z$ approaches $S_j$) 
is derived from the 
above-mentioned fact  that the logarithm in \eqref{eq:chi_j} vanishes 
at $\tau=S_j$. 

The well-definedness of  $\chi_j(S_j)$ has been explained above. 
The points $T_j$'s have been chosen just for 
simplicity, not for necessity, 
and can be replaced by any other points on 
$\arc{S_1S_2}$ or $\arc{S_3S_4}$  
(it results in another decomposition of $\delta$).   
Since there is nothing special about them 
as far as the product $\delta(z)$ is concerned, 
it is  well-defined and bounded on each $V_j$. 
Hence $\delta(z)$, $\delta^{-1}(z)$ and their boundary values
are bounded everywhere.

\begin{remark}
\label{rem:chidelta}
If $j\ne k$, then $r(-\tau)=-r(\tau)$ implies 
\(  \chi_{j+2}(S_{j+2}) =\chi_j  (S_j) \),
 \;
\(  \delta_{j+2}(S_{k+2})=\delta_{j}(S_{k})\) 
for $j=1, 2$ and $k=1, 2, 3, 4$ 
with the convention $S_{5}=S_1, S_6=S_2$.
\end{remark}

\subsection{Scaling operators}

In a neighborhood of each saddle point $S_j$, 
we have 
\begin{equation}
 \varphi(z)=\varphi(S_j)+\frac{\varphi''(S_j)}{2} (z-S_j)^2 + \varphi_j(z)
 , \quad
 \varphi_j(z)=O(|z-S_j|^3)
\label{eq:varphi_j}
\end{equation}
for some   function $\varphi_j(z)$. 
Set 
\[\beta_j=(-1)^j 2^{-1}  i (4t^2-n^2)^{-1/4}S_j \]
so that  $\varphi''(S_j)\beta_j^2=(-1)^{j-1} i/2$. See \eqref{eq:N_jexp}.

Let the \textit{infinite} crosses $\Sigma(S_j)$'s and $\Sigma(0)_j$'s be defined by 
\begin{align*}
 \Sigma(S_j)
 &=(S_j+A\mathbb{R})\cup (S_j+ i A\mathbb{R}), 
 \;\mbox{oriented inward  } &  \; (j=1, 3), 
 \\
 &=(S_j+\bar A\mathbb{R})\cup (S_j+ i \bar A\mathbb{R}), 
 \;\mbox{oriented outward  } & \; (j=2, 4), 
 \\
 \Sigma(0)_j
 &= e^{\pi i/4}\mathbb{R} \cup  e^{-\pi i/4}\mathbb{R}, 
 \;\mbox{oriented inward } &  \; (j=1, 3),
 \\
 &=\mbox{the same set as above, }
 \;\mbox{oriented outward } &  \; (j=2, 4).
\end{align*}
Note that 
$\Sigma(S_j)$ is obtained by extending one of the four small crosses 
forming $\Sigma'$ 
and  that the cut of $\delta_j$, namely $\arc{S_jT_j}$,  is 
between two rays of $\Sigma(S_j)$.

\begin{figure}\centering
  \includegraphics[width=3cm]{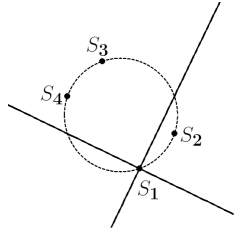}
 \caption{$\Sigma(S_1)$: an infinite cross}
 \label{fig:SigmaS1}
\end{figure}

We introduce the scaling operators with rotation 
\begin{align*}
 N_j\colon
 &
 (\mathcal{C}^0 \cup L^2)(\Sigma(S_j))\to
 (\mathcal{C}^0  \cup L^2)(\Sigma(0)_j), \quad
 \\ & f(z)\mapsto 
 (N_j f)(z)=f(\beta_j z+S_j).
\end{align*}
It is the pull-back by the mapping 
$M_j\colon \Sigma(0)_j \to \Sigma(S_j)$, 
$z\mapsto  \beta_j z+S_j$. 
The real axis  is  mapped by $M_j$ to the tangent line of 
the circle at $S_j$. 
The positive imaginary axis  is mapped  
to the outer ($j=1, 3$) or inner ($j=2, 4$) normal.  
Hence the singularity of $\delta(z)$ is,  if seen through $M_j$, to the left of $0$ 
either $j$ is even or odd.  

By (\ref{eq:delta_j}) and  (\ref{eq:varphi_j}) we have
\begin{align}
  (N_j \delta_j)(z)
 &= \left(
         \frac{\beta_j z}{\beta_j z+S_j-T_j}
      \right)^{  (-1)^{j-1}  i \nu_j} 
          e^{  (-1)^{j-1} N_j \chi_j(z)} 
 \nonumber 
 \\
 &=
  \left(
   \frac{\beta_j}{S_j-T_j}
  \right)^{  (-1)^{j-1}  i \nu_j}
  z^{  (-1)^{j-1}  i \nu_j}
  \left(
   \frac{S_j-T_j}{\beta_j z+S_j-T_j}
  \right)^{  (-1)^{j-1} i \nu_j} 
    \nonumber
  \\
  &\hspace{1.5em}\times
   e^{  (-1)^{j-1}\chi_j(S_j)}  
   e^{  (-1)^{j-1}\{N_j \chi_j(z)-\chi_j(S_j)\}}   . 
  \label{eq:N_jdelta_j}
  \\
  (N_j  e^{-\varphi})(z)
  &=S_j^n  e^{- i t (S_j-S_j^{-1})^2/2} 
    e^{(-1)^j  i z^2/4} 
    e^{- (N_j\varphi_j)(z)}.
 \label{eq:N_jexp}
\end{align}
Here the arguments of 
$\beta_j/(S_j-T_j)$ and of 
$(S_j-T_j)/(\beta_j z+S_j-T_j)$ (at least for a large positive $z$) 
are between $-\pi/2$ and $\pi/2$.

Originally, $N_j\delta_j$ has a cut along the preimage under $M_j$ of 
$\arc{S_jT_j}$. 
It is an arc%
\footnote{Here we are assuming $0<n\le (2-V_0)t$. 
If $-(2-V_0)t \le n \le 0$, the central angle is not less than $\pi/4$ 
but is less than $\pi/2$. The  consideration about 
the homotopic movement requires no change 
at all. }  
with central angle not exceeding $\pi/4$
which is tangent to the real line at the endpoint $0=M_j^{-1}(S_j)$.  
See Figure~\ref{fig:cut}. 
It  is  in the region $\pi\le \arg z \le 5\pi/4$  (if $j$ is odd) or  $3\pi/4\le \arg z \le  \pi$  (if $j$ is even). 
The factor 
$z^{ i \nu_j}$ is originally cut along the union of the preimage and 
the half-line   
$C_j\colon z(u)=\beta_j^{-1}(-S_j+T_j)-u, u\in\mathbb{R}_+$.    
See Figure~\ref{fig:cut}. 
Since we consider $z^{ i \nu_j}$ only on the 
cross $\Sigma(0)_j$, the cut can be moved homotopically 
as long as it is away from the cross. 
Hence the cut of $z^{ i \nu_j}$ is moved to $\mathbb{R}_-$ 
for $j=1, 2, 3, 4$.   
Another factor 
$\{(S_j-T_j)/(\beta_j z+S_j-T_j) \}^{ i \nu_j}$ is cut 
along the half-line $C_j$, but its singularity eventually disappears 
as $t\to\infty, \beta_j\to 0$.

\begin{figure}\centering
  \includegraphics[width=3cm]{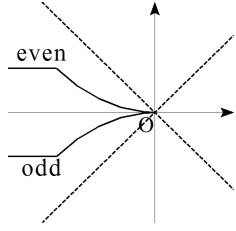}
 \caption{Cut of $N_j \delta_j$ and the cross}
 \label{fig:cut}
\end{figure}

Set 
\begin{equation}
    \widehat\delta_j(z)=\delta(z)/\delta_j(z)=\prod_{k\ne j}\delta_k(z).
   \label{eq:hatdeltajdef}
\end{equation}
By (\ref{eq:N_jdelta_j}) and (\ref{eq:N_jexp}), we have
\begin{equation}
 N_j [\delta  e^{-\varphi}]=N_j [\delta  e^{-i t \psi}](z)
 =\delta_j^0 \delta_j^1(z), 
\end{equation}
where 
\begin{align}
 & \delta_j^0 
 = S_j^n  e^{- i t (S_j-S_j^{-1})^2/2} 
      \left( \frac{\beta_j}{S_j-T_j} \right)^{  (-1)^{j-1} i \nu_j}                 
         e^{  (-1)^{j-1} \chi_j(S_j)}       \widehat\delta_j (S_j), 
 \label{eq:deltaj0defdecember}
  \\ 
  &  \delta_j^1(z) 
  = 
    z^{  (-1)^{j-1}  i \nu_j}
    e^{(-1)^j  i z^2/4 }
    \left(
   \frac{S_j-T_j}{\beta_j z+S_j-T_j}
  \right)^{  (-1)^{j-1}  i \nu_j} 
  \nonumber
  \\
  &\hspace{5em}\times
    e^{  (-1)^{j-1}\{N_j\chi_j(z)-\chi_j(S_j)\} - (N_j\varphi_j)(z)}  
    \frac{(N_j  \widehat\delta_j) (z)}{ \widehat\delta_j (S_j)}.
\end{align}
We choose the branch of the logarithm which is  real on 
$\mathbb{R}_+$ and 
cut along $\mathbb{R}_-$. The imaginary powers 
in  the definitions of  $\delta^0_j$ and $\delta^1_j(z)$ should be  
interpreted accordingly. 
We have at least a pointwise convergence 
$\delta_j^1(z)\to z^{(-1)^{j-1}  i \nu_j}  e^{(-1)^j  i z^2/4}$ 
as $t \to \infty$.
In Proposition~\ref{prop:localization} we shall show that 
this convergence is uniform in a certain sense. 
  
\begin{remark}
We have 
 \(
   \hat\delta_{j+2}(S_{j+2})=\hat\delta_j(S_j)
   ,\; 
    \delta_{j+2}^0
       =(-1)^n \delta_j^0
    \; (j=1, 2)
 \) 
because of Remark~\ref{rem:chidelta}. 
Since $T_1=T_2=\exp(-\pi i/4)$,  
the ratio  $\beta_j/(S_j-T_j)$ 
 in the definition of $\delta^0_j$ can be represented by
\begin{equation}
  \frac{\beta_1}{S_1-T_1}=
 \frac{- i A}{2(4t^2-n^2)^{1/4} (A-1)}, \quad
  \frac{\beta_2}{S_2-T_2}=
 \frac{  i \bar A}{2(4t^2-n^2)^{1/4} (\bar A-1)}.
 \label{eq:D1D2december}
\end{equation} 
These quantities are nothing but $D_1$ and $D_2$ in 
(\ref{eq:D1D2}). 
\end{remark}

The cross $\Sigma(S_j)$ is the union of the two lines 
$\Sigma(S_j, L)$ and $\Sigma(S_j, L')$ defined by 
\begin{align*}
 \Sigma(S_j, L)&=S_j+  e^{\pi i/4}S_j \mathbb{R}
                    =S_j+A \mathbb{R}
 &\quad  (j=1, 3) ,
 \\
                    &=S_j+  e^{-\pi i/4}S_j \mathbb{R}
                      =S_j +  i \bar{A} \mathbb{R}
 &\quad  (j=2, 4) ,
 \\
  \Sigma(S_j, L')&=S_j+  e^{-\pi i/4}S_j \mathbb{R}
                       =S_j +  i A  \mathbb{R}
 &\quad  (j=1, 3) ,
 \\
                    &=S_j+  e^{\pi i/4}S_j \mathbb{R}
                     =S_j +   \bar{A} \mathbb{R}
 &\quad  (j=2, 4) .
\end{align*}
Notice that each $\Sigma(S_j, L)$ [resp. $\Sigma(S_j, L')$] 
share some segment  with $L$ [resp. $L'$], but is not included in it. 
We have 
$M_j^{-1} \Sigma(S_j, L)= e^{-\pi i/4}\mathbb{R} \,(j=1, 3)$, 
\textit{the direction of steepest descent} of $ e^{- i z^2/4}$, 
and  
$M_j^{-1} \Sigma(S_j, L)= e^{\pi i/4}\mathbb{R} \,(j=2, 4)$, 
the direction of steepest descent of $ e^{ i z^2/4}$. 
Therefore  $\Sigma(S_j, L)$ is in the direction of 
steepest descent of $ e^{-\varphi}$ 
for any $j=1, 2, 3, 4$ 
in view of (\ref{eq:N_jexp}).
Similarly,  $\Sigma(S_j, L')\,(j=1, 2, 3, 4)$ is in the direction of 
steepest descent of $ e^{\varphi}$.

We introduce some sets each of which is  in a neighborhood of a saddle point. 
(In contrast, $L_\varepsilon$ and $L'_\varepsilon$ to be introduced later 
are away from the saddle points. )
\begin{align*}
   \Sigma(S_j, L)_\varepsilon
  & =\{z\in \Sigma(S_j, L); |z-S_j|\le\varepsilon\}
  \; &\mbox{(a short segment)}
   , 
   \\
    \Sigma(S_j, L')_\varepsilon
  & =\{z\in \Sigma(S_j, L'); |z-S_j|\le\varepsilon\}
    \; &\mbox{(another short segment)}
    , 
    \\
   \Sigma(S_j)_\varepsilon
  & =   \Sigma(S_j, L)_\varepsilon \cup   \Sigma(S_j, L')_\varepsilon 
  \; &\mbox{(a small cross)}. 
\end{align*} 
Then 
\(
    \Sigma(S_j)_\varepsilon=
  \{z \in \Sigma(S_j); \, |z-S_j|\le\varepsilon\}
\), 
\(
 \Sigma'=\cup_{j=1}^4 \Sigma(S_j)_\varepsilon
\), 
\( \mathrm{supp}\, w'_+\subset \Sigma'\cap L
  =\cup_{j=1}^4 
 \Sigma(S_j, L)_\varepsilon
 \) and 
\( \mathrm{supp}\, w'_-\subset \Sigma'\cap L' 
 = \cup_{j=1}^4 
 \Sigma(S_j, L')_\varepsilon
\).

\section{Crosses}
\label{sec:crosses}

Split $\Sigma'$ into the union of four disjoint small crosses: 
$\Sigma'=\cup_{j=1}^4 \Sigma(S_j)_\varepsilon$. 
We decompose $w'_\pm$ into the form 
\begin{equation}
 w'_\pm =\sum_{j=1}^4  w^j_\pm, 
 \label{eq:w'decomposition}
\end{equation}
where 
\(
 \mathrm{supp\,}  w^j_+ \subset \Sigma(S_j, L)_\varepsilon
\) and 
\(
 \mathrm{supp\,}  w^j_-  
 \subset \Sigma(S_j, L')_\varepsilon
\). 
Set 
\(
  w^j= w^j_+  +  w^j_- 
\). 
Define the operators 
$A_j=C_{w^j}\,(j=1, 2, 3, 4)$  on 
$L^2(\Sigma')$ as in (\ref{eq:integopdef}). 
We have 
$w'=\sum_j w^j$ and 
$C_{w'}^{\Sigma'}=\sum_j A_j$. 
See Proposition~\ref{prop:Qnw'2}  for the distinction between 
$C^{\Sigma'}_{w'}$ and $C_{w'}$. 

\begin{lemma}
If $j\ne k$, we have 
\begin{equation}
 \|
   A_j A_k
 \|_{L^2(\Sigma')}
 \le Ct^{-1/2}, \quad 
 \|
   A_j A_k
 \|_{L^\infty(\Sigma') \to L^2(\Sigma')}
 \le Ct^{-3/4}.
\end{equation}
\label{lem:AjAkestimates}
\end{lemma}

\begin{proof}
Follow the proof of \cite[Lemma~3.5]{DZ}.  
Use Lemma~\ref{lem:wawbwcw'estimates} and the fact that 
$\mathrm{dist}(\Sigma(S_j)_\varepsilon, \Sigma(S_k)_\varepsilon )$ 
is bounded from below. 
\end{proof}

In \S\ref{sec:boundedness} we will prove the existence and boundedness of 
\( 
  (1-A_j)^{-1}\colon L^2(\Sigma')\to L^2(\Sigma')\, 
  (j=1, 2, 3, 4)
\).       
In view of Lemma~\ref{lem:AjAkestimates} 
 and \cite[Lemma~3.15]{DZ}, it leads to that of 
\(
   (1-C_{w'}^{\Sigma'})^{-1}\colon L^2(\Sigma') \to L^2(\Sigma')
 \). 
By \cite[Lemma~2.56]{DZ}, we find that 
\(
   (1-C_{w'})^{-1}\colon L^2(\Sigma)\to L^2(\Sigma)
\) exists and  is bounded. 
Finally  
\( 
   (1-C_{w^\sharp})^{-1}
   \colon L^2(\Sigma) \to L^2(\Sigma)
\)
 also exists  and is  bounded 
because of the smallness of 
$C_{w^\sharp}-C_{w^\prime}=C_{w^e}$ and the second resolvent 
identity $(1-A)^{-1}=\{(1-B)^{-1}(B-A) +1\}^{-1}(1-B)^{-1} $.

We have
\begin{align}
 & 
 \|A_j A_k (1-A_k)^{-1} I   \|_{L^2(\Sigma')}
 \nonumber
 \\
 & \le 
 \|A_j A_k I\|_{L^2(\Sigma')} +
 \|A_j A_k (1-A_k)^{-1} A_k I\|_{L^2(\Sigma')}. 
\end{align}
Here $ \|A_j A_k I\|_{L^2(\Sigma')} \le Ct^{-3/4}$ follows from 
the second inequality of   Lemma~\ref{lem:AjAkestimates}. 
On the other hand, we have 
\(
 \|A_k I\|_{L^2(\Sigma')} \le \|w^k\|_{L^2(\Sigma')} 
  \le  \|w'\|_{L^2(\Sigma')} \le Ct^{-1/4}
\) 
by (\ref{eq:w'L2L2estimates}). 
This estimate, together with the first inequality of  Lemma~\ref{lem:AjAkestimates} and the 
boundedness of $(1-A_j)^{-1}$,  yields 
\begin{equation}
 \|A_j A_k (1-A_k)^{-1} I   \|_{L^2(\Sigma')}
 \le Ct^{-3/4}.
 \label{eq:20111127}
\end{equation}
Following (repeatedly)   \cite[pp.338-339]{DZ}, 
we obtain by \eqref{eq:20111127} and Lemma~\ref{lem:AjAkestimates} 
\begin{align}
& \int_{\Sigma'} z^{-2}
 \left( (1-C_{w'}^{\Sigma'})^{-1} I\right)
 w'
 \nonumber
 \\
& =
 \int_{\Sigma'} z^{-2}
 \Bigl(  
    \bigl\{
      {\textstyle   1+\sum_j A_j (1-A_j)^{-1}  }
    \bigr\}
 I\Bigr) w'
 +O(t^{-1})
 \nonumber
 \\
 &=
 \int_{\Sigma'} z^{-2} w' 
  +\sum_{j, k=1}^4 z^{-2} 
 \left(
  A_j (1-A_j)^{-1} I
 \right)
 w^k
 +O(t^{-1}).
 \label{eq:1004_1}
\end{align}
We have sixteen 
quantities involving the pairs  
$(A_j (1-A_j)^{-1}, w^k)\,(j, k\in\{1, 2, 3,  4\})$. We claim that the main contribution 
is by the `diagonal' pairs ($j=k$). Let us estimate 
the `off diagonal' terms. 

Since $(1-A_j)^{-1}=1+(1-A_j)^{-1} A_j$, we have
\begin{align}
 &
 \biggl|
   \int_{\Sigma'} z^{-2}\bigl(A_j (1-A_j)^{-1} I\bigr)w^k
 \biggr| 
 \nonumber
 \\
 &\le 
 \biggl| 
   \int_{\Sigma'} z^{-2} (A_j I) w^k 
 \biggr|
 +
  \biggl| 
   \int_{\Sigma'} z^{-2}
   \bigl\{A_j (1-A_j)^{-1}  A_j     I \bigr\} 
   w^k 
 \biggr|.
 \label{eq:offdiag}
\end{align}
Notice that 
the distance of    $z\in\Sigma(S_k)_\varepsilon$ and 
 $\eta\in\Sigma(S_j)_\varepsilon$ and that of  
$z\in\Sigma'$ and $0$ are  bounded from below. These facts, 
 combined with (\ref{eq:w'L2L2estimates2}),  
 lead to the following Fubini type estimate of the iterated integral in 
 the first term: 
\begin{align}
 \qquad{}&
 \biggl| 
   \int_{\Sigma'} z^{-2} (A_j I) w^k 
 \biggr|
 =
 \biggl| 
   \int_{\Sigma (S_k)_\varepsilon} z^{-2} 
    \biggl(  
      \int_{\Sigma (S_j)_\varepsilon} 
      \frac{w^j (\eta)}{\eta-z}
       \frac{\, d\eta}  {2\pi i}
    \biggr)
   w^k (z) \, d z
 \biggr|
 \nonumber
 \\
 {}\qquad &\le 
 C\|w^j (\eta)\|_{L^1(\Sigma(S_j)_\varepsilon)}
 \|w^k (z)\|_{L^1(\Sigma(S_k)_\varepsilon)}
 \le Ct^{-1}.
 \label{eq:Fubini}
\end{align} 
The second term in (\ref{eq:offdiag}) is estimated in 
a slightly different way. 
By (\ref{eq:w'L2L2estimates}),  (\ref{eq:w'L2L2estimates2}) 
and the Schwarz inequality, 
\begin{align}
 {}\qquad&
  \biggl| 
   \int_{\Sigma'} z^{-2}
   \bigl\{A_j (1-A_j)^{-1}  A_j     I \bigr\} 
   w^k 
 \biggr|
 \nonumber
 \\
 {}\qquad& = 
  \biggl| 
   \int_{\Sigma (S_k)_\varepsilon} 
   \Bigl\{
     \int_{\Sigma (S_j)_\varepsilon} 
     \frac{\bigl((1-A_j)^{-1} A_j I\bigr)(\eta) w^j(\eta)}
     {\eta-z}  
     \frac{\, d \eta}{2\pi i}
   \Bigr\} 
   z^{-2} w^k (z)
 \biggr|
 \nonumber
 \\
 {}\qquad& \le C\|(1-A_j)^{-1} A_j I\|_{L^2 (\Sigma(S_j)_\varepsilon)}
 \|
   w^j (\eta)
 \|_{L^2 (\Sigma(S_j)_\varepsilon)}
 \|
   z^{-2} w^k (z)
 \|_{L^1 (\Sigma(S_k)_\varepsilon)}
 \nonumber
 \\
 {} \qquad& 
 \le C \|w^j\|^2_{L^2 (\Sigma(S_j)_\varepsilon)} 
 \|
     w^k 
 \|_{L^1 (\Sigma(S_k)_\varepsilon)}
 \le Ct^{-1}. 
 \label{eq:Fubini2}
\end{align}
By using (\ref{eq:1004_1})-(\ref{eq:Fubini2}) 
and 
\(
 w'=\sum_{j=1}^4 \left( (1-A_j) (1-A_j)^{-1} I\right) w^j
\),  we   get
\begin{align*}
 &  
 \int_{\Sigma'} z^{-2}
 \left( (1-C_{w'}^{\Sigma'})^{-1} I\right)
 w'
 =
 \sum_{j=1}^4 \int_{\Sigma(S_j)_\varepsilon}
 z^{-2} \bigl[(1-A_j)^{-1} I\bigr] w^j
 +O(t^{-1})   .
\end{align*}
Owing to  \cite[(2.61)]{DZ}, 
\( 
 A_j=C_{w^j}\colon L^2(\Sigma')\to L^2(\Sigma')
\) 
in the above formula 
can be replaced by 
\(
 C^\varepsilon_{w^j}\colon L^2 (\Sigma(S_j)_\varepsilon)\to  L^2 (\Sigma(S_j)_\varepsilon)
\), 
which is defined as in (\ref{eq:integopdef}) for the pair 
$(\Sigma(S_j)_\varepsilon, w^j)$.
Combining this fact with 	Proposition~\ref{prop:Qnw'2} ($\ell\ge 1$), 
we get the following result, which shows that 
the contributions of the four small crosses can be 
separated out. 
\begin{proposition} We have
\[{}\qquad
  R_n(t)
  =-\delta(0) \sum_{j=1}^4 \int_{\Sigma(S_j)_\varepsilon} z^{-2}
  \Bigl[
   \bigl(      (1-C^\varepsilon_{w^j})^{-1}                 I\bigr) 
    (z)
   w^j  (z)
  \Bigr]_{21}  
  \frac{\, d z}{2\pi i}
  +O(t^{-1}).
\]
\label{prop:Rn4crosses}
\end{proposition}

\section{Infinite crosses and localization}
\label{sec:infinitecrosses}

We introduce  $\hat w^j_\pm$ and $\hat w^j$ on  
the infinite cross $\Sigma(S_j)$ given by
\begin{align}
 \hat w^j_\pm 
  &  =w^j_\pm, \quad & z\in \Sigma(S_j)_\varepsilon, 
   \nonumber 
   \\
  \hat w^j_\pm   &  =0, \quad & z\in \Sigma(S_j)\setminus\Sigma(S_j)_\varepsilon,  
   \nonumber 
   \\
   \hat w^j  
  &  = \hat w^j_+ + \hat w^j_-, \quad & z\in \Sigma(S_j).
  \nonumber
\end{align}
Define the operator $\hat A_j\colon L^2(\Sigma(S_j))\to L^2(\Sigma(S_j))$ as 
in (\ref{eq:integopdef}) with the kernel $\hat w^j_\pm$. 
Set 
\(
   \Delta_j^0=(\delta_j^0)^{\sigma_3}, \, 
   \tilde \Delta_j^0\phi=\phi\Delta_j^0
\). 
The operator $\tilde \Delta_j^0$ and its inverse are  bounded. 
Define 
$\alpha_j\colon  L^2(\Sigma(0)_j)\to L^2(\Sigma(0)_j)$ by 
\begin{align}
 {}\qquad
  \alpha_j
 &=C_{(\Delta_j^0)^{-1} (N_j \hat w^j) \Delta_j^0 } 
 \nonumber
 \\
{} &=C_+ (\bullet\, (\Delta_j^0)^{-1} (N_j \hat w^j_-) \Delta_j^0)
    +C_- (\bullet\, (\Delta_j^0)^{-1} (N_j \hat w^j_+) \Delta_j^0). 
\label{eq:alpha_j}
\end{align}
Then we have 
\begin{equation}
 \alpha_j= \tilde\Delta_j^0 N_j \hat A^j N_j^{-1} (\tilde\Delta_j^0)^{-1}, 
 \quad
 \hat A^j=N_j^{-1}  (\tilde\Delta_j^0)^{-1} \alpha_j  \tilde \Delta_j^0 N_j.
\label{eq:alpha_j-scaling}
\end{equation}

On 
$M_j^{-1}\Sigma(S_j, L)_\varepsilon\setminus\{0\}\subset\Sigma(0)_j$, 
we have
\[ 
  \Bigl( (\Delta_j^0)^{-1} (N_j [z^{-2}\hat w^j_+]) \Delta_j^0 \Bigr)(z)
 =
 \begin{bmatrix} 0 \quad & (\beta_j z+S_j)^{-2}  R(\beta_j z+S_j) 
 \delta_j^1 (z)^2 \\ 0 \quad & 0
 \end{bmatrix}
\]
by \eqref{eq:w'+},  
and the left-hand side is zero  on
 $\Sigma(0)_j\setminus M_j^{-1}\Sigma(S_j, L)_\varepsilon$. 
 
On 
$M_j^{-1}\Sigma(S_j, L')_\varepsilon  \subset\Sigma(0)_j$,
we have 
\[ 
  \Bigl( (\Delta_j^0)^{-1} (N_j [z^{-2}\hat w^j_-]) \Delta_j^0 \Bigr)(z)
 =
 \begin{bmatrix} 0\quad  &  0 \\ 
 -(\beta_j z+S_j)^{-2} \bar R(\beta_j z+S_j) \delta_j^1 (z)^{-2} \quad  & 0
 \end{bmatrix}
\]
by \eqref{eq:w'-}, 
and the left-hand side is zero  on
 $\Sigma(0)_j\setminus M_j^{-1}\Sigma(S_j, L')_\varepsilon$.

Assume that  $j$ is odd ($j=1, 3$). We have 
\(
 M_j^{-1} \Sigma(S_j, L)_\varepsilon
   =  e^{-\pi i/4}\mathbb{R}\cap\{|\beta_j z|\le\varepsilon\}
\) and it is the union of  the lower right  ($z/ e^{-\pi i/4}\ge 0$) 
and upper left ($-z/ e^{-\pi i/4}\ge 0$)   parts. 
Recall that the positive imaginary axis is mapped by $M_j$ 
to the \textit{outer} normal at $S_j$ if $j$ is odd. 

\begin{proposition} 
Assume $j=1$ or $3$. 
Fix an arbitrary constant $\gamma$ with $0<2\gamma<1$. 
Then on 
\(
 M_j^{-1} \Sigma(S_j, L)_\varepsilon
 \cap
\{z; \, \pm z/ e^{-\pi i/4}>0\}
\) respectively,  we have
\begin{align}
 &
  \Bigl|    (\beta_j z +S_j)^{-2}
   R(\beta_j z +S_j) \delta_j^1 (z)^2
   - S_j^{-2} R(S_j \pm) z^{2 i \nu_j}  e^{- i z^2/2}
 \Bigr|
 \nonumber 
 \\
 &\le 
 C t^{-1/2}\log t \cdot \Bigl| e^{-\frac{ i \gamma z^2}{2}}  \Bigr|, 
 \label{eq:difficultineq}
 \\
  &
  \Bigl|  
   R(\beta_j z +S_j) \delta_j^1 (z)^2
   -  R(S_j \pm) z^{2 i \nu_j}  e^{- i z^2/2}
 \Bigr|
 \nonumber 
 \\
 &
 \le 
 C t^{-1/2}\log t \cdot \Bigl| e^{-\frac{ i \gamma z^2}{2}}  \Bigr|.
 \label{eq:easyineq}
\end{align}
We choose    $R(S_j+)=\bar r(S_j)$ on the lower right part 
and $R(S_j-)=-\bar r(S_j)/(1-|r(S_j)|^2)$  on the upper left part. 
The analytic functions  $z^{\pm 2 i \nu_j}$  and   $\delta^1_j(z)$  
have  cuts along  
$\mathbb{R}_-$. 
\label{prop:localization}
\end{proposition}
\begin{proof} We show  only \eqref{eq:difficultineq} 
because \eqref{eq:easyineq} is just an easier version of it. Moreover we can assume 
$j=1$ by symmetry. 
On 
\(
 M_1^{-1} \Sigma(S_1, L)_\varepsilon
\),  we have
\begin{align}
 &
 \Bigl| e^{\frac{ i \gamma z^2}{2}}  \Bigr|
  \Bigl|  
   (\beta_1 z +S_1)^{-2}
   R(\beta_1 z +S_1) \delta_1^1 (z)^2
   -S_1^{-2} R(S_1 \pm) z^{2 i \nu_1}  e^{- i z^2/2}
 \Bigr|
 \nonumber
 \\
 &\le 
 \Bigl| 
  e^{-\frac{ i \gamma z^2}{2}} 
 \Bigl[ 
     R (\beta_1 z +S_1)  
     \mathrm{F}      \mathrm{E}_1   \mathrm{E}_2 
    \frac{N_1 \hat\delta_1 (z)^2}{\hat\delta_1 (S_1)^2}
    -
      S_1^{-2} R(S_1\pm) z^{2 i \nu_1} 
     e^{ i (-1+2\gamma)\frac{z^2}{2}}
 \Bigr]
 \Bigr|, 
 \label{eq:manyfactors}
\end{align}
where
\begin{align}
  \mathrm{F}
  &=
    (\beta_1 z +S_1)^{-2}
  \Bigl(\frac{S_1-T_1}{\beta_1 z+S_1-T_1}\Bigr)^{2 i \nu_1}, 
   \\
  \mathrm{E}_1
  &= z^{2 i \nu_1}  \exp
        \biggl( i (-1+2\gamma)\frac{z^2}{2}
          -2N_1 \varphi_1(z)        
        \biggr), 
   \\
  \mathrm{E}_2
  &=\exp
  \Bigl( 2\{     N_1\chi_1(z)-\chi_1(S_1)  \}
  \Bigr).
\end{align}
Each factor in (\ref{eq:manyfactors}) is uniformly bounded 
with respect to $(n, t)$.

Since $z  e^{- i \gamma z^2/2}$ is bounded and 
$\beta_1=O(t^{-1/2})$, we have
\begin{align}
 & 
  \bigl|
   e^{-\frac{ i \gamma z^2}{2}}
  \bigl[
    R(\beta_1 z +S_1) - R(S_1 \pm)
  \bigr]
 \bigr|
 \nonumber
 \\
 &\le 
\bigl|  e^{-\frac{ i \gamma z^2}{2}}\bigr|  |\beta_1 z|
 \sup_{|w|\le\varepsilon}|R'(w+S_1)| 
 \le C |\beta_1|\le C t^{-1/2}.
 \label{eq:Restimate}
\end{align}
Similarly, we obtain 
\begin{align}
 &|
   e^{-\frac{ i \gamma z^2}{2}}
  (
    F -S_1^{-2}
  )
 |
 \le C t^{-1/2}, 
 \label{eq:FracEstimate}
 \\
 & \biggl|
      e^{-\frac{ i \gamma z^2}{2}}
    \biggl[
     \frac{N_1 \hat\delta_1 (z)^2}{\hat\delta_1 (S_1)^2}    -1
    \biggr]
   \biggr|
    \le C t^{-1/2}. 
 \label{eq:delta1Estimate}
\end{align}
Moreover, $N_1\varphi_1(z)=O(|\beta_1 z|^3)$  
and the boundedness of $z^{2 i \nu_1}$ and $z^3  e^{- i \gamma z^2/2}$ imply 
\begin{align}
 {} \qquad &\left|
    e^{- i\gamma z^2/2}
   \left[
     \mathrm{E}_1
     - z^{2 i \nu_1} \exp\left( i(-1+2\gamma)z^2/2\right)
   \right]
 \right|
 \nonumber
 \\
 {}&\le
 \left|  
     e^{- i\gamma z^2/2}  
 \right| 
 \sup_{0\le s \le 1}
 \left|
  \frac{ \, d}{\, d s} 
   \exp\biggl( i (-1+2\gamma)\frac{z^2}{2}
               -2s N_1 \varphi_1 (z)
       \biggr)
 \right|
 \nonumber
 \\
 {}&\le
 C 
 \left|  
     e^{- i\gamma z^2/2} 
 \right| 
 |\beta_1 z|^3
 \le  C |\beta_1|^3 \le Ct^{-3/2}\le Ct^{-1/2}.
 \label{eq:E1estimate}
\end{align}
Four $t^{-1/2}$ type estimates (\ref{eq:Restimate})-(\ref{eq:E1estimate}) have been obtained. 

Lastly we derive a $t^{-1/2}\log t$  type estimate involving $E_2$. 
We have  
\begin{align}
  &
  \left|
     e^{- i\gamma z^2/2}
   (
   \mathrm{E}_2
       -1
    )
 \right|
 \nonumber
 \\
 &\le 
 \sup_{0\le s \le 1}
 \left|
    e^{  2s\{N_1\chi_1 (z)-\chi_1(S_1)   \}  }
 \right|
 \cdot
 \left|
   2 e^{- i \gamma z^2/2}
   (N_1\chi_1 (z)-\chi_1(S_1) )
 \right| .
 \label{eq:E2}
\end{align}
Since the supremum is bounded, we have only to derive a 
$t^{-1/2}\log t$  type estimate of the second factor in (\ref{eq:E2}).
Integration by parts yields
\begin{align}
 & 2\pi  i
 [N_1 \chi_1 (z)-\chi_1(S_1)]
 \nonumber
 \\
 &=\int_{T_1}^{S_1} 
   \log \frac{1-|r(\tau)|^2}{1-|r(S_1)|^2}
   \, d \log \frac{\tau-(\beta_1 z+S_1)}{\tau-S_1}
 =-L_1-L_2, 
 \label{eq:L1andL2}
\end{align}
where 
\begin{align}
 & L_1 
 = \log  \frac{1-|r(T_1)|^2}{1-|r(S_1)|^2}
      \log\frac{T_1-(\beta_1 z+S_1)}{T_1-S_1}, 
 \\
  &L_2=\int_{T_1}^{S_1} 
      \log\frac{\tau-(\beta_1 z+S_1)}{\tau-S_1}
      g(\tau)\, d \tau, 
 \\
 & g(\tau)=\frac{\, d}{\, d \tau}\log (1-|r(\tau)|^2). 
\end{align}
The first logarithm in $L_1$ is bounded. 
The second logarithm can be estimated in the same way 
as (\ref{eq:Restimate}) etc. and we get  
\begin{equation}
 | e^{- i \gamma z^2/2} L_1|\le Ct^{-1/2}. 
 \label{eq:L1estimate}
\end{equation}
We express the   integral $L_2$  as the sum of  two terms:
\begin{align}
 L_2
 &=L_2^1+L_2^2, 
 \nonumber
 \\
 L_2^1&=
 \int_{T_1}^{S_1} \{g(\tau)-g(S_1)\} 
 \log \frac{\tau-(\beta_1 z+S_1)}{\tau-S_1}\, d\tau, 
 \\
 L_2^2
 &=
 g(S_1)\int_{T_1}^{S_1}
 \log \frac{\tau-(\beta_1 z+S_1)}{\tau-S_1}\, d\tau. 
\end{align}
We have 
\(
 |\log(1+w)|=|\int_0^w (1+z)^{-1} \,\, d z|\le C|w|
\) 
in any sector  that is away from  
the negative real axis. For $\tau$ and $z$ in $L_2^1$, 
the ratio $-\beta_1 z/(\tau-S_1)$ is in such a sector and 
\begin{equation}
 \left|
  	\log \frac{\tau-(\beta_1 z+S_1)}{\tau-S_1}
 \right|
 =
  \left|
  	\log \Bigl( 1- \frac{\beta_1 z}{\tau-S_1} \Bigr)
 \right|
 \le 
 C \left|   \frac{\beta_1 z}{\tau-S_1}\right|.
\end{equation}
It implies 
\begin{equation}
 {}\qquad | e^{- i\gamma z^2/2}  L_2^1|
   \le C| e^{- i\gamma z^2/2} |  |\beta_1 z| 
 \int_{T_1}^{S_1} 
 \left|
    \frac{g(\tau)-g(S_1)}{\tau-S_1}
 \right|
 |\, d \tau|
 \le Ct^{-1/2}. 
  \label{eq:L21estimate}
\end{equation}

Next we consider $L^2_2$. 
An elementary calculation shows 
\begin{align}
 {}\qquad
 &\int_{T_1}^{S_1} \log \frac{\tau-(\beta_1 z+S_1)}{\tau-S_1}\, d\tau
 \nonumber 
 \\
 {} \qquad
 &=
 \Bigl[
  (\tau-\beta_1 z-S_1)    \log ( \tau-\beta_1 z-S_1)  
  - 
  (\tau-S_1)\log (\tau-S_1)
 \Bigr]_{\tau=T_1}^{\tau=S_1}
 \nonumber
 \\
  {} \qquad
 &= -\beta_1 z\log(-\beta_1 z)
   -\Bigl[
      (T_1-S_1-w)\log (T_1-S_1-w)
     \Bigr]_{w=0}^{w=\beta_1z}.
  \label{eq:L22}
\end{align}
The product of $ e^{- i \gamma z^2/2}$ and the second term in the right-hand side of 
(\ref{eq:L22})  can be dealt with in the usual way, as in 
(\ref{eq:Restimate}) etc. 
The product of $ e^{- i \gamma z^2/2}$ and the first term   
enjoys an estimate involving $t^{-1/2} \log t$. 
Indeed,  if $t$ is sufficiently large,   
\begin{align}
 | e^{- i\gamma z^2/2} \beta_1 z \log(-\beta_1 z)|
 &\le
 | e^{- i\gamma z^2/2} \beta_1 z \log  z|
 +| e^{- i\gamma z^2/2}z|| \beta_1  \log(-\beta_1 )|
 \nonumber
 \\
 & \le Ct^{-1/2}+Ct^{-1/2}\log t
 \le Ct^{-1/2}\log t.
\end{align}
By \eqref{eq:L22} and these estimates, we get 
\begin{equation}
 |
  e^{-i \gamma z^2/2} L^2_2
 |
 \le C t^{-1/2} \log t . 
 \label{eq:L22bis}
\end{equation}
Combining \eqref{eq:E2}, \eqref{eq:L1andL2}, \eqref{eq:L1estimate},  
\eqref{eq:L21estimate} and \eqref{eq:L22bis}, we  obtain 
\begin{equation}
 | e^{- i\gamma z^2/2} (\mathrm{E}_2-1)|
  \le Ct^{-1/2}\log t. 
 \label{eq:E2estimate}
\end{equation}
Finally   \eqref{eq:difficultineq} follows from 
(\ref{eq:manyfactors}), (\ref{eq:Restimate}), (\ref{eq:FracEstimate}), 
(\ref{eq:delta1Estimate}), (\ref{eq:E1estimate}) and  (\ref{eq:E2estimate}). 
\end{proof}

If $j$ is odd ($j=1, 3$),  we have 
\(
 M_j^{-1} \Sigma(S_j, L')_\varepsilon
   =  e^{\pi i/4}\mathbb{R}\cap\{|\beta_j z|\le\varepsilon\}
\) and it is the union of  the upper right 
  ($z/ e^{\pi i/4}\ge 0$) 
and lower left ($-z/ e^{\pi i/4}\ge 0$)   parts. 

\begin{proposition} 
\label{prop:localization2} 
Assume $j=1$ or $3$. 
Fix an arbitrary constant $\gamma$ with $0<2\gamma<1$. 
Then on 
\(
 M_j^{-1} \Sigma(S_j, L')_\varepsilon
 \cap
\{z; \, \pm z/ e^{\pi i/4}>0\}
\),  we have
\begin{align}
 &
  \Bigl|    (\beta_j z +S_j)^{-2}
   \bar R(\beta_j z +S_j) \delta_j^1 (z)^{-2}
   - S_j^{-2} \bar R(S_j \pm) z^{-2 i \nu_j}  e^{ i z^2/2}
 \Bigr|
 \nonumber 
 \\
 &\le 
 C t^{-1/2}\log t \cdot \Bigl| e^{\frac{ i \gamma z^2}{2}}  \Bigr|, 
 \label{eq:difficultineq2}
 \\
  &
  \Bigl|  
   \bar  R(\beta_j z +S_j) \delta_j^1 (z)^{-2}
   -  \bar  R(S_j \pm) z^{-2 i \nu_j}  e^{ i z^2/2}
 \Bigr|
 \nonumber 
 \\
 &
 \le 
 C t^{-1/2}\log t \cdot \Bigl| e^{\frac{ i \gamma z^2}{2}}  \Bigr|.
 \label{eq:easyineq2}
\end{align}
We choose    $\bar R(S_j+)= r(S_j)$ on the upper right part 
and $\bar R(S_j-)=-  r(S_j)/(1-|r(S_j)|^2)$  on the lower left part. 
\label{prop:localizationL'}
\end{proposition}

\begin{remark}
\label{rem:S2S4}
If $j$ is even, we have 
\(
 M_j^{-1} \Sigma(S_j, L)_\varepsilon
   =  e^{\pi i/4}\mathbb{R}\cap\{|\beta_j z|\le\varepsilon\}
\)
and 
\(
 M_j^{-1} \Sigma(S_j, L')_\varepsilon
   =  e^{-\pi i/4}\mathbb{R}\cap\{|\beta_j z|\le\varepsilon\}
\).  
The positive imaginary axis is mapped by $M_j$ to the 
\textit{inner} normal at $S_j$. 
Roughly speaking,  we have the following:
\begin{itemize}
 \item  On $M_j^{-1} \Sigma(S_j, L)_\varepsilon$ ($j=2, 4$); \par\noindent
 	$(\beta_j z +S_j)^{-2}   R(\beta_j z +S_j) \delta_j^1 (z)^2$ tends to 
      $S_j^{-2} R(S_j\pm) z^{-2 i \nu_j}  e^{ i z^2/2}$ and 
       $R(\beta_j z +S_j) \delta_j^1 (z)^2$ tends to 
      $R(S_j\pm) z^{-2 i \nu_j}  e^{ i z^2/2}$ as $t\to\infty$.  
     Here we choose   $R(S_j+)=\bar r(S_j)$ on the upper right part 
     and $R(S_j-)=-\bar r(S_j)/(1-|r(S_j)|^2)$ on the lower left part. 
 \item  On $M_j^{-1} \Sigma(S_j, L')_\varepsilon$  ($j=2, 4$); \par\noindent
 	$(\beta_j z +S_j)^{-2}   \bar R(\beta_j z +S_j) \delta_j^1 (z)^{-2}$ tends to 
      $S_j^{-2} \bar R(S_j\pm) z^{2 i \nu_j}  e^{- i z^2/2}$ and 
       $\bar R(\beta_j z +S_j) \delta_j^1 (z)^2$ tends to 
      $\bar R(S_j\pm) z^{-2 i \nu_j}  e^{ i z^2/2}$ as $t\to\infty$.  
     Here we choose 
     $\bar R(S_j+)=  r(S_j)$ on the lower right part and 
      $\bar R(S_j-)=- r(S_j)/(1-|r(S_j)|^2)$ on the upper left part.
\end{itemize}
\end{remark}

\section{Boundedness of inverses}  
\label{sec:boundedness}
Recall that $A_j$ is an operator on $L^2(\Sigma')$ with 
the kernel $w^j$ supported by $\Sigma(S_j)_\varepsilon$ and that  
$\Sigma'=\cup_{j=1}^4 \Sigma(S_j)_\varepsilon$. 
In this section, we prove that   
 $(1-A_j)^{-1}$ exists and is bounded as an operator on
  $L^2(\Sigma')$. 
This fact was used in \S\ref{sec:crosses}. 
We make three steps of reduction (which will be followed by still 
other steps later in this section). 
It is enough to prove the existence and boundedness of: 
\begin{enumerate}
\item 
\(
   (1-C^\varepsilon_{w^j})^{-1}
   \colon 
   L^2(\Sigma(S_j)_\varepsilon) \to  L^2(\Sigma(S_j)_\varepsilon)
\). 
\item
$(1-\hat A_j)^{-1}\colon  L^2(\Sigma(S_j))\to  L^2(\Sigma(S_j))$.  
\item
$(1-\alpha_j)^{-1}\colon L^2(\Sigma(0)_j)\to L^2(\Sigma(0)_j)$. 
\end{enumerate}
The first two steps of reduction are due to \cite[Lemma 2.56]{DZ}. 
The third is due to a scaling argument. 
Indeed,  (\ref{eq:alpha_j-scaling}) implies 
\begin{equation}
  (1- \hat A^j)^{-1}
 =N_j^{-1}  (\tilde\Delta_j^0)^{-1} (1-\alpha_j)^{-1}  
 \tilde \Delta_j^0 N_j
\label{eq:DZ(3.68)}
\end{equation}
and the boundedness of $(1- \hat A^j)^{-1}$ follows from 
that of $(1-\alpha_j)^{-1}$, 
since $N_j, \tilde\Delta_j^0$ and their inverses are bounded. 

\medskip

Set 
\begin{equation}
 \omega^j_\pm = 
 (\Delta_j^0)^{-1} (N_j \hat w^j_\pm) \Delta_j^0  
 =
 (\delta_j^0)^{-\mathrm{ad\,}\sigma_3} N_j \hat w^j_\pm, 
\label{eq:omegajdef}
\end{equation}
so that by \eqref{eq:alpha_j}
\[
 \alpha_j=C_{\omega_j}
 =C_+ (\bullet\, \omega^j_-)+C_- (\bullet\, \omega^j_+). 
\]
The cross $\Sigma(0)_j$ consists of four rays: 
\[
  \Sigma(0)_j=\cup_{k=1}^4   \Sigma(0)_j^k, \quad 
   \Sigma(0)_j^k=
     e^{ i s_k \pi/4} \mathbb{R}_+ , 
\]
where $s_1=1, s_2=3, s_3=5, s_4=7$. 
Each ray is oriented inward if $j$ is odd and outward  if $j$ is even. Set
\begin{equation}
 \omega^{j, \infty}_\pm=\lim_{t\to\infty}\omega^j_\pm. 
 \label{eq:omegaconvergence}
\end{equation}
By (\ref{eq:easyineq}), \eqref{eq:easyineq2} and 
Remark~\ref{rem:S2S4},  \eqref{eq:omegaconvergence} 
holds in $L^p\, (1\le p\le \infty)$.  The concrete forms 
of $\omega^{j, \infty}_\pm$ are given below.

\subsection{Case A } 
\label{subsec:CaseA} 
Assume that $j$ is odd ($j=1, 3$). 
The contour $\Sigma(0)_j$ is oriented inward. 
Notice that $r(S_1)=-r(S_3), \nu_1=\nu_3$. 
By virtue of (\ref{eq:w'+}), (\ref{eq:w'decomposition}) and (\ref{eq:omegajdef})  we get
\begin{align*}
  \omega^{j, \infty}_+
 &= 
 \begin{bmatrix} 0 & \quad 
 -\frac{\bar r(S_j)}{  1-|r(S_j)|^2 } z^{2 i\nu_j} e^{- i z^2/2} 
  \\ 0   & \quad  0 
 \end{bmatrix} , 
 \; &z\in \Sigma(0)^2_j,  
 \\ 
  {}
 &= \begin{bmatrix} 0 & \quad  \bar r(S_j) z^{2 i\nu_j}  e^{- i z^2/2}
 \\ 0 & \quad   0\end{bmatrix} , 
 \; &z\in \Sigma(0)^4_j,  
\\
   {}
 &= 0 , 
  \; &z\in \Sigma(0)^1_j \cup \Sigma(0)^3_j.    
\end{align*}
and  (\ref{eq:w'-}), (\ref{eq:w'decomposition}) and (\ref{eq:omegajdef})  imply  
\begin{align*}
 \omega^{j, \infty}_-
 &= 
 \begin{bmatrix} 0 & \quad 0 
  \\ -r(S_j) z^{-2 i\nu_j}  e^{ i z^2/2}&   \quad  0
\end{bmatrix} , 
 &\quad z\in \Sigma(0)^1_j, 
 \\
  {}
 &= 
 \begin{bmatrix} 0 & \quad 0 \\ 
   \frac{r(S_j)}{1-|r(S_j)|^2} z^{-2 i\nu_j} e^{ i z^2/2}
   & \quad 0 
 \end{bmatrix} , 
 &\quad z\in \Sigma(0)^3_j, 
 \\
  {}
 &= 0 , 
  &\quad z\in \Sigma(0)^2_j \cup \Sigma(0)^4_j.
\end{align*}
For each $j$, either $\omega_+^{j, \infty}$ or  $\omega_-^{j, \infty}$ is $0$ and 
the associated jump matrix   is 
\begin{equation}
 (I-\omega_-^{j, \infty})^{-1} (I+\omega_+^{j, \infty})
= (I+\omega_-^{j, \infty}) (I+\omega_+^{j, \infty})
  =I+\omega_+^{j, \infty} \mbox{ or }  I+\omega_-^{j, \infty}.
\label{eq:jumpmatrix}
\end{equation}

Set 
\(
 \omega^{j, \infty}
   =\omega^{j, \infty}_+ + \omega^{j, \infty}_-
\) and 
\[
       \alpha^{j, \infty}=
      C_{\omega^{j, \infty}}
      =C_+(\bullet\,\omega^{j, \infty}_-)+C_-(\bullet\, \omega^{j,    
      \infty}_+)
      \colon L^2 (\Sigma(0)_j) \to  L^2 (\Sigma(0)_j)
      .
\]
By (\ref{eq:easyineq}), (\ref{eq:easyineq2}) and 
Remark~\ref{rem:S2S4}, we find that 
the boundedness of $(1-\alpha_j)^{-1}$ can be derived from that of 
\(
   (1-\alpha^{j, \infty})^{-1}      
   \colon L^2 (\Sigma(0)_j) \to  L^2 (\Sigma(0)_j)
\) 
if $t$ is sufficiently large. 
The proof of the boundedness of $(1-\alpha^{j, \infty})^{-1}$ 
(at least for $j=1, 3$) can 
be found in \cite{Important, DZUTYO}. 
Indeed, the matrices 
$\omega^{j, \infty}_\pm$ are  the same 
(up to inversion in the case of different orientations) as those in 
\cite[p.198]{Important}  and  \cite[p.46]{DZUTYO}. 
The  presentation in the former is   sketchy. 
A complete proof  is given in the latter, but probably 
it is not easy to find, especially at libraries outside Japan.  So here we repeat key steps of 
the calculation in  \cite{DZUTYO}. 
The method is basically the same as that in \cite{DZ}, which can be 
referred to for some details. 

Reorient and extend $\Sigma(0)_j$ to $\Sigma^{e}$ (we do without 
the subscript $j$ for simplicity) which is defined as follows:
\begin{itemize}
\item
 \(\Sigma^{e}=\Sigma(0)_j \cup \mathbb{R}
   =\mathbb{R}\cup  e^{\pi  i/4}\mathbb{R} \cup 
     e^{-\pi  i/4}\mathbb{R}
 \) as sets. 
\item
 $\mathbb{R}$ is unconventionally oriented from the  right to the left, 
 $ e^{\pi  i/4}\mathbb{R}$ from the lower left to the upper right, 
 $ e^{-\pi  i/4}\mathbb{R}$ from the upper left to the 
 lower right. 
\end{itemize}  
We have $\Sigma^{e}=\cup_{k=1}^6 \Sigma^{e}_k$, 
$\Sigma^{e}_k= e^{ i a_k \pi/4}\mathbb{R}_+$ 
($a_1=1, a_2=3, a_3=4, a_4=5, a_5=7, a_6=8$). 
Let $\Omega^{e}_k$($k=1, 2, \dots, 6$) be the sector between 
$\Sigma^{e}_{k-1}$ and $\Sigma^{e}_k$ (here $\Sigma^{e}_0=\Sigma^{e}_6$). 
The rays are so oriented that 
$C_+^{\Sigma^{e}} f$ and $C_-^{\Sigma^{e}} f$ are  analytic 
in the even- and odd-numbered sectors  respectively  
for $f\in L^2(\Sigma^{e})$.

From $\omega^{j, \infty}_\pm$, we obtain the renewed 
jump matrix 
\[v^e=v^e(z)=(b^e_-)^{-1} b^e_+
 =(I-\omega^e_-)^{-1} (I+\omega^e_+)
  =(I+\omega^e_-)  (I+\omega^e_+),
\]
where 
\begin{align*}
  b^e_+ &= B_1=
  \begin{bmatrix} 1 & \quad 0 \\ r(S_j)z^{-2 i \nu_j} e^{ i z^2/2} & \quad 1
  \end{bmatrix}
  , b^e_-=I 
  \;  &\mbox{ on }\Sigma^{e}_1, 
  \\
   b^e_+&= B_2=
  \begin{bmatrix} 1 &  
   \quad  -\frac{\bar r(S_j)}{1-|r(S_j)|^2}   z^{2 i \nu_j} e^{- i z^2/2} 
  \\ 0 & \quad  1
  \end{bmatrix}
  , b^e_-=I 
  \;  &\mbox{ on }\Sigma^{e}_2,  
  \\
    b^e_\pm &=I    \;  &\mbox{ on }\Sigma^{e}_3, 
  \\
  b^e_+&=I, 
  b^e_-=  B_3=
  \begin{bmatrix} 1 \quad & 0 \\
    -\frac{r(S_j)}{ 1-|r(S_j)|^2}   
           z^{-2 i \nu_j} e^{ i z^2/2} \quad 
    &  1
  \end{bmatrix} 
  \;  &\mbox{ on }\Sigma^{e}_4,  
  \\
  b^e_+&=I, 
  b^e_-=  B_4=
    \begin{bmatrix} 1 & \quad 
   \bar  r(S_j)z^{2 i \nu_j} e^{- i z^2/2} 
   \\ 0 & \quad 1
  \end{bmatrix}
  \;  &\mbox{ on }\Sigma^{e}_5, 
  \\
    b^e_\pm &=I   \;  &\mbox{ on }\Sigma^{e}_6, 
   \\
   \omega^e_\pm &=   \pm (b^e_\pm -I) 
    \;  &\mbox{ on }\Sigma^{e}. 
\end{align*}
Notice that we have changed orientations on 
$\Sigma^{e}_1 \cup \Sigma^{e}_5=\Sigma(0)_j^1 \cup \Sigma(0)_j^4$, 
where $\omega^e_\pm=-\omega^{j, \infty}_\mp$. 
The operator $\alpha^{j, \infty}=C_{\omega^{j, \infty}}$ of type (\ref{eq:integopdef}) associated with 
$(\Sigma(0)_j, \omega^{j, \infty}_\pm)$ coincides with 
that associated with 
\((\Sigma^{e}_1 \cup \Sigma^{e}_2 \cup \Sigma^{e}_4 \cup \Sigma^{e}_5, 
  \omega^e_\pm
   |_{\Sigma^{e}_1 \cup \Sigma^{e}_2 \cup \Sigma^{e}_4 \cup \Sigma^{e}_5}
 )
\). 
By \cite[Lemma 2.56]{DZ}, 
in order to prove the boundedness of 
$(1-\alpha^{j, \infty})^{-1}\colon L^2(\Sigma(0)_j)\to L^2(\Sigma(0)_j)$, 
we have only to prove that of 
$(1-C_{\omega^e})^{-1}$, where $C_{\omega^e}$ is
the operator 
associated with $(\Sigma^{e}, \omega^e_\pm)$.

Define a piecewise analytic matrix function $\phi(z)$ on 
$\mathbb{C}\setminus\Sigma^{e}$ as follows:
\begin{equation}
 \begin{array}{|c||c|c|c|c|c|c|}
 \hline  \mbox{}&
  \Omega^{e}_1 & \Omega^{e}_2 & \Omega^{e}_3 & 
  \Omega^{e}_4 & \Omega^{e}_5 & \Omega^{e}_6 
  \\
  \hline 
  \phi(z)&
  z^{- i\nu_j\sigma_3}B_1^{-1}
  & z^{- i\nu_j\sigma_3}
  & z^{- i\nu_j\sigma_3} B_2^{-1}
  & z^{- i\nu_j\sigma_3} B_3^{-1}
  & z^{- i\nu_j\sigma_3}
  & z^{- i\nu_j\sigma_3} B_4^{-1}
  \\
  \hline
 \end{array}
 \label{eq:phitable}
\end{equation}
On   $\Sigma^{e}$, set 
\(
   v^{e, \phi}(z)=\phi_-(z) v^e(z) \phi_+^{-1}(z)
\). 
Then we have $v^{e, \phi}(z)=I$ on 
\(
   \Sigma^{e}_1 \cup \Sigma^{e}_2 \cup \Sigma^{e}_4 \cup \Sigma^{e}_5
\). 
On $\Sigma^{e}_6$, we have 
\begin{align*}
 v^{e, \phi}(z)
    &=(z^{- i\nu_j\sigma_3} B_1^{-1})
       (z^{- i\nu_j\sigma_3} B_4^{-1})^{-1}
    =z^{- i\nu_j  \mathrm{ad}\sigma_3}(B_1^{-1}B_4)
    \\
    &= e^{-\frac{ i z^2}{4}  \mathrm{ad}\sigma_3}
                   \begin{bmatrix}
                          1   & \quad \bar r(S_j) \\ -r(S_j) & 	 \quad 1-|r(S_j)|^2 
                   \end{bmatrix}.
\end{align*}
On $\Sigma^{e}_3=\mathbb{R}_-$, its orientation implies 
$(z^{-i\nu_j})_- / (z^{-i\nu_j})_+ =  e^{2\pi \nu_j }=(1-|r(S_j)|^2)^{-1}$ 
and  
\begin{align*}
 v^{e, \phi}(z)
 &= e^{-\frac{ i z^2}{4}  \mathrm{ad}\sigma_3}
   \left\{
     \begin{bmatrix}	
      1 & \frac{\bar r(S_j)}{1-|r(S_j)|^2} \\ 0 &  1 
     \end{bmatrix}
      (z^{-i\nu_j \sigma_3})_- (z^{i\nu_j \sigma_3})_+ 
     \begin{bmatrix}	
      1 &0 \\ \frac{- r(S_j)}{1-|r(S_j)|^2}  &  1 
     \end{bmatrix}
   \right\}
 \\
 &=
    e^{-\frac{ i z^2}{4}  \mathrm{ad}\sigma_3}
                   \begin{bmatrix}
                          1   & \bar r(S_j) \\ -r(S_j) &  1-|r(S_j)|^2 
                   \end{bmatrix}.
\end{align*}

We set 
\begin{align}
  & \omega^{e, \phi}_\pm=0 \mbox{ on }  
  \Sigma^{e}_1 \cup \Sigma^{e}_2 \cup \Sigma^{e}_4 \cup \Sigma^{e}_5, 
 \\
 & 
    \omega^{e, \phi}_- = 
    \begin{bmatrix} 0&\quad 0\\-r(S_j) e^{ i z^2/2} &\quad  0 \end{bmatrix}, \;
     \omega^{e, \phi}_+ =
    \begin{bmatrix} 0&\quad  \bar r(S_j) e^{- i z^2/2}\\0 & \quad 0 \end{bmatrix}
     \mbox{ on }   \mathbb{R}, 
  \\
 &
      \omega^{e, \phi}= \omega^{e, \phi}_+ +  \omega^{e, \phi}_-  
         \mbox{ on }  \Sigma^{e}.
\end{align}
We have $v^{e, \phi}=(I- \omega^{e, \phi}_-)^{-1}(I+ \omega^{e, \phi}_+)$ 
on $\Sigma^e$.
By \cite[Lemma 2.56]{DZ}, 
the boundedness of $(1-C_{\omega^{e, \phi}})^{-1}$ on 
$L^2 (\Sigma^e)$ follows from that of 
$(1-C_{\omega^{e, \phi}}|_\mathbb{R})^{-1}$ on 
$L^2 (\mathbb{R})$. 
The latter is an immediate consequence of $|r(S_j)|<1$. 

By means of the process of \cite[pp.344-346]{DZ}
and the several steps of reduction in this section,  
we can derive the boundedness of 
$(1-C_{\omega^e})^{-1}$, 
$(1-\alpha^{j, \infty})^{-1}=(1-C_{\omega^{j, \infty}})^{-1}$, 
 $(1-\alpha^{j})^{-1}=(1-C_{\omega^j})^{-1}$, $(1-\hat A_j)^{-1}$, 
$(1-C^\varepsilon_{w^j})^{-1}$ and $(1-A_j)^{-1}$.

\subsection{Case B}  Assume that $j$ is even ($j=2, 4$). 
The contour $\Sigma(0)_j$ is oriented outward. 
Notice that $r(S_2)=-r(S_4), \nu_2=\nu_4$. 
We have 
\begin{align*}
  \omega^{j, \infty}_+
 &=\begin{bmatrix} 0 & \quad   \bar r(S_j) z^{-2 i\nu_j}  e^{ i z^2/2}
  \\ 0   & \quad  0 \end{bmatrix} , 
 &\quad z\in \Sigma(0)^1_j, 
 \\
  \omega^{j, \infty}_+
 &=\begin{bmatrix} 0 
 & \quad  -\frac{\bar r(S_j)}{  1-|r(S_j)|^2 }z^{-2 i\nu_j} e^{ i z^2/2}
 \\ 0 & \quad  0\end{bmatrix} , 
 &\quad z\in \Sigma(0)^3_j, 
\\
   \omega^{j, \infty}_+
 &= 0 , 
  &\quad z\in \Sigma(0)^2_j \cup \Sigma(0)^4_j.  
\end{align*}
and 
\begin{align*}
 \omega^{j, \infty}_-
 &=\begin{bmatrix} 0 & \quad 0 \\ 
 \frac{r(S_j)}{ 1-|r(S_j)|^2 } z^{2 i\nu_j} e^{- i z^2/2}
  & \quad 0 \end{bmatrix} , 
 &\quad z\in \Sigma(0)^2_j, 
 \\
  \omega^{j, \infty}_-
 &=\begin{bmatrix} 0 & \quad 0 \\ -r(S_j) z^{2 i\nu_j}  e^{- i z^2/2}&    \quad 
 0\end{bmatrix} , 
 &\quad z\in \Sigma(0)^4_j, 
 \\
  \omega^{j, \infty}_-
 &= 0 , 
  &\quad z\in \Sigma(0)^1_j \cup \Sigma(0)^3_j.
\end{align*}

Define 
\(
 \omega^{j, \infty}
\) and 
\(
       \alpha^{j, \infty}
\) in the same way as in Case A. 
Here again we want to show the boundedness of 
\(
   (1-\alpha^{j, \infty})^{-1}      
   \colon L^2 (\Sigma(0)_j) \to  L^2 (\Sigma(0)_j)
\).

Denote 
 $\omega_\pm^{j, \infty}$ 
by 
\(
  \omega_\pm^{\mathrm{even}, \infty}   
\) 
or 
\(
  \omega_\pm^{\mathrm{odd}, \infty}   
\) 
when $j$ is   even or odd respectively. 
Replace $r(S_j)$ with $\bar r(S_j)$ (hence 
 $\bar r(S_j)$ with $r(S_j)$) in the definition of 
$\omega_\pm^{\mathrm{odd}, \infty}$, and  
denote by $\bar \omega_\pm^{\mathrm{odd}, \infty}$ 
the matrix thus obtained. 
For example, when $j=2$, we have 
\[
  \bar\omega^{\mathrm{odd}, \infty}_+ =0, \quad
  \bar\omega^{\mathrm{odd}, \infty}_-
  =\begin{bmatrix}
     0 &\quad 0
     \\
     -\bar r(S_2) z^{-2i \nu_2} e^{iz^2/2} & \quad 0 
    \end{bmatrix}
    \quad
    \mbox{on }
    \Sigma(0)^1_2. 
\] 
Notice that $\bar r$ is evaluated at $j=2$ (an even number), 
although  the form of the matrix is borrowed from the 
odd $j$'s and the superscript contains the word `odd'.

The problem concerning  
$\bar \omega_\pm^{\mathrm{odd}, \infty}$ can be solved 
in the same way as that concerning 
$\omega_\pm^{\mathrm{odd}, \infty}$. 
We find that 
\begin{equation}
 \omega_\pm^{\mathrm{even}, \infty}
 =-{}^t \bar\omega_\mp^{\mathrm{odd}, \infty}
 =- \sigma_1 \bar \omega_\mp^{\mathrm{odd}, \infty} \sigma_1 , 
 \quad 
  \sigma_1=\begin{bmatrix} 0& 1 \\ 1& 0\end{bmatrix}. 
\label{eq:sigma1}
\end{equation}
We denote  $\alpha^{j, \infty}=C_{\omega^{j, \infty}}$  by 
$\alpha^{\mbox{\scriptsize{even}}, \infty}$ or 
$\alpha^{\mbox{\scriptsize{odd}}, \infty}$ when 
$j$ is even or odd respectively, and let 
$\bar \alpha^{\mathrm{odd}, \infty}$ be the 
operator obtained by replacing 
$\omega_\pm^{\mathrm{odd}, \infty}$ 
with $\bar \omega_\pm^{\mathrm{odd}, \infty}$ 
in $\alpha^{\mbox{\scriptsize{odd}}, \infty}$. 
Let $\hat \sigma_1$ be the right multiplication by $\sigma_1$, namely  
\(
 \hat \sigma_1 (f)=f\sigma_1 
\). 
Then (\ref{eq:sigma1}) implies 
\begin{equation}
 \alpha^{\mathrm{even}, \infty}
 = \hat \sigma_1 \circ
      \bar \alpha^{\mathrm{odd}, \infty}
   \circ   \hat \sigma_1, 
\end{equation}
because the cross $\Sigma(0)_j$ changes orientation 
in accordance with the parity of $j$. 
It  cancels out the negative sign in the 
right-hand side of  (\ref{eq:sigma1}). 
Moreover it exchanges the positive and negative sides of the rays, 
as is compatible with the $\pm$ and $\mp$ signs in  
(\ref{eq:sigma1}). 
Therefore the boundedness of 
$(1-\alpha^{\mathrm{odd}, \infty})^{-1}$  proved in the 
previous subsection   
implies that of  $(1-\bar \alpha^{\mathrm{odd}, \infty})^{-1}$  
and $(1-\alpha^{\mathrm{even}, \infty})^{-1}$.

\section{Reconstruction via scaling}
\label{sec:reconstscaling}
\subsection{Reduction to infinity}

We defined the operator 
\(
 \hat A_j \colon L^2(\Sigma(S_j)) \to  L^2(\Sigma(S_j))
\) 
with the kernel $\hat w^j_\pm$ 
at the beginning of \S\ref{sec:infinitecrosses}. 
By Proposition~\ref{prop:Rn4crosses} and \cite[Lemma 2.56]{DZ}, 
we obtain 
\begin{equation}
 R_n(t)
 = -\delta(0) \sum_{j=1}^4 \left[ R_n^j (t)\right]_{21}+O(t^{-1}), 
 \label{eq:Rnj1}
\end{equation}
where the \textit{matrix} $R_n^j(t)$ is defined by
\begin{equation}
  R_n^j (t)
 =  \int_{\Sigma(S_j)}
   \bigl((1-\hat A_j)^{-1}I\bigr)(z)  z^{-2} \hat w^j (z)
 \frac{\, d z}{2\pi  i}  .
 \label{eq:Rnj2}
\end{equation}
By (\ref{eq:DZ(3.68)}) and  $\tilde\Delta^0_j  N_j   I=\Delta^0_j$,  
we obtain
\begin{align*}
 R_n^j(t)
 &= \int_{\Sigma(S_j)}
   \bigl( N_j^{-1} (\tilde\Delta^0_j)^{-1} (1-\alpha_j)^{-1}  
   \tilde\Delta^0_j 
     N_j   I
   \bigr)(z)  z^{-2} \hat w^j (z)
   \frac{\, d z}{2\pi  i} 
  \\ 
  & 
  =  \int_{\Sigma(S_j)}
     \bigl((1-\alpha_j)^{-1}\Delta_j^0  \bigr) (z')
     (\Delta_j^0)^{-1}
     \bigl(N_j[\bullet^{-2} \hat w^j]\bigr)(z')
  \frac{\, d z}{2\pi  i}  
 \\
   &   \mbox{with} \; z'=M_j^{-1}(z)=(z-S_j)/\beta_j
  \in \Sigma(0)_j .
\end{align*}
The change of variables $z=\beta_j z'+S_j$ leads to 
\begin{equation*}
  R^j_n(t)
 =  \beta_j \int_{\Sigma(0)_j} 
     \bigl((1-\alpha_j)^{-1}\Delta_j^0  \bigr) (z')
     (\Delta_j^0)^{-1}
     \bigl(N_j[\bullet^{-2} \hat w^j]\bigr)(z')    
  \frac{\, d z'}{2\pi  i} .
\end{equation*}
The operator $\alpha_j$ is  
basically right action. It commutes with the left multiplication 
 by $\Delta_j^0=(\delta_j^0)^{\sigma_3}$ and  so does $(1-\alpha_j)^{-1}$. 
We get by (\ref{eq:omegajdef})
\begin{equation}
 R^j_n(t)
 = \beta_j  
  (\delta_j^0)^{\mathrm{ad\,}\sigma_3}
 \int_{\Sigma(0)_j} 
     \bigl((1-\alpha_j)^{-1} I \bigr) (z)
      (\beta_j z+S_j)^{-2} \omega^j (z)    
       \frac{\, d z}{2\pi  i}  .
  \label{eq:RnjSigma0j}
\end{equation}
By using Proposition~\ref{prop:localization}, \ref{prop:localization2} 
and Remark~\ref{rem:S2S4}, 
we get
\begin{align}
 & \int_{\Sigma(0)_j} 
     \bigl((1-\alpha_j)^{-1} I \bigr) (z)
      (\beta_j z+S_j)^{-2} \omega^j (z)    
 \, d z 
 \nonumber
 \\
  & =\int_{\Sigma(0)_j} 
     \bigl((1-\alpha_j^\infty)^{-1} I \bigr) (z)
      S_j^{-2} \omega^{j, \infty} (z)    
 \, d z   
 +O(t^{-1/2}\log t).
 \label{eq:reducetoinfty}
\end{align}
We substitute (\ref{eq:reducetoinfty}) into (\ref{eq:RnjSigma0j}). 
Then (\ref{eq:Rnj1}) and $\beta_j=O(t^{-1/2})$ yield  the following proposition. 
\begin{proposition}
\label{prop:1025}
\begin{align}
& R_n(t)
 =-\frac{\delta(0)}{2\pi i}\sum_{j=1}^4 \beta_j (\delta_j^0)^{-2}
  S_j^{-2}
  \left[
  \int_{\Sigma(0)_j}   
     \bigl((1-\alpha_j^\infty)^{-1} I \bigr) (z)
     \omega^{j, \infty} (z)    
     \, d z   
  \right]_{21}
  \nonumber
  \\
 &\hspace{4em}+O(t^{-1}\log t). 
 \label{eq:Rnt1024}
\end{align}
\end{proposition}

The integral in (\ref{eq:Rnt1024}) can be calculated by 
using a Riemann-Hilbert problem. 
For $\mathbb{C}\setminus \Sigma(0)_j$, set
\[
 m^j(z)=I+\int_{\Sigma(0)_j} 
              \frac{\left((1-\alpha_j^\infty)^{-1} I\right)(\zeta)  
              \omega^{j, \infty}(\zeta)}
              {\zeta -z}
              \frac{\, d \zeta}{2\pi  i}.
\]
Then $m^j(z)$ solves (uniquely) the Riemann-Hilbert problem
\begin{align}
  & m^j_+(z)=m^j_-(z)v^j(z), & z\in\Sigma(0)_j, 
  \\
  & m^j(z)\to I                & \mbox{as}\;z\to\infty,
\end{align}
with 
\(v^j(z)
  = (I-\omega_-^{j, \infty})^{-1}  (I+\omega_+^{j, \infty}).
\)
As $z\to\infty$, $m^j(z)$ behaves like 
$m^j(z)=I-z^{-1} m_1^j +O(z^{-2})$, where
\begin{equation}
 m_1^j= \int_{\Sigma(0)_j} 
         \left((1-\alpha_j^\infty)^{-1} I  \right)
          (\zeta)  
              \omega^{j, \infty}(\zeta)
         \frac{\, d \zeta}{2\pi  i}
\end{equation}
is nothing but the integral in (\ref{eq:Rnt1024}). It implies the 
following proposition.

\begin{proposition}
\label{prop:1030}
We have 
\begin{equation}
 R_n(t)
 =-\frac{\delta(0)}{2\pi i}\sum_{j=1}^4 \beta_j (\delta_j^0)^{-2}
  S_j^{-2}
 (m^j_1)_{21}
  +O(t^{-1}\log t). 
 \label{eq:Rnt1030}
\end{equation}
\end{proposition}

The integral $m_1^j$ is calculated in two steps depending 
on the parity of $j$. 

\subsection{Case A} 
Assume that   $j$  is odd. 
We introduce the contour $\hat \Sigma$ in the following way:
\begin{itemize}
 \item $\hat \Sigma=\Sigma^e$   as sets. 
 \item 
Set $\hat\Sigma_k=\Sigma^e_k$, $\hat\Omega_k=\Omega^e_k$ 
for $k=1, 2, \dots, 6$.  Orient each   $\hat\Sigma_k$ 
from the (upper/lower) left to the  (upper/lower) right. 
In other words, the orientation of $\hat\Sigma$ differs from that of 
$\Sigma^e$   on $\mathbb{R}$. 
\end{itemize}
Set 
\(
  H=m^j \phi^{-1} 
\), 
where $\phi$ is as in \eqref{eq:phitable}.   
Its jump matrix  on $\hat \Sigma\setminus\mathbb{R}$ 
is 
\(
 \phi_- v^j \phi_+^{-1}=\phi_- v^e \phi_+^{-1}=v^{e, \phi}=I
\). 
Denote by $\hat v$ its jump matrix on 
$\mathbb{R}=\Sigma^e_3 \cup \Sigma^e_6$.  We have 
\begin{align}
 & H_+ =H_- \hat v  \mbox{ on } \mathbb{R},    
 \label{eq:H_RH}
 \\
 & 
 H z^{- i \nu_j \sigma_3} 
 =I-z^{-1}m_1^j +O(z^{-2}) 
   \mbox{ as } z\to \infty.
   \label{eq:hatPsi_asymp}
\end{align}
We can show that  
\[
    \hat v =
   v^{e, \phi}(z)^{-1}|_{\mathbb{R}}
   = e^{- \frac{i z^2}{4} \mathrm{ad} \sigma_3}
     \begin{bmatrix}    1-|r(S_j)|^2 & \quad -\bar r(S_j) \\ r(S_j) & \quad  1
   \end{bmatrix}.
\]

\begin{remark}
\label{rem:rSj=0}
If $r(S_j)=0$, then we have $\nu_j=0, \hat v=I$. It follows 
that $H=I, m_1^j=0$. Hence $M_j$ in  Theorem~\ref{thm:main} 
vanishes. 
\end{remark}

Assume $r(S_j)\ne 0$. 
Our formulas (\ref{eq:H_RH}) and 
(\ref{eq:hatPsi_asymp}) are just the counterparts of 
\cite[(4.18),  (4.19)]{DZ}. 
We can follow the calculation in 
\cite[pp.349-352]{DZ}   and get 
\begin{equation}
  (m_1^j)_{21}
  =- \frac{ i\sqrt{2\pi}  e^{- i \pi/4}   e^{-\pi\nu_j/2}  }
            {\bar r(S_j) \Gamma( i \nu_j)}, 
     \quad
  (m_1^j)_{12}
  =\frac{ i \sqrt{2\pi}  e^{ i \pi/4}   e^{-\pi\nu_j/2}  }
            { r(S_j) \Gamma(- i \nu_j)}.
  \label{eq:mj112}
\end{equation}
It follows from $r(S_1)=-r(S_3)$ that 
\[
  (m_1^3)_{21}=- (m_1^1)_{21},    \quad (m_1^3)_{12}=- (m_1^1)_{12}. 
\]

\subsection{Case B }
Before calculating $m_1^j$ when $j$  is even,     
we give a general argument. 
Let us consider a pair of Riemann-Hilbert problems on a 
common contour: 
\begin{align}
 & M_+=M_- v, \quad 
 M \to I \mbox{  as  } z\to \infty, 
 \label{eq:RHpair1}
 \\
 & \tilde M_+=\tilde M_- (\sigma_1 v \sigma_1), \quad 
 \tilde M \to I \mbox{  as  } z\to \infty. 
 \label{eq:RHpair2}
\end{align}
The latter implies  
\[
   \sigma_1 \tilde M_+ \sigma_1
   =
  (\sigma_1 \tilde M_- \sigma_1)v, 
  \quad 
  \sigma_1 \tilde M \sigma_1\to I \mbox{  as  } z\to \infty
\]
and $\sigma_1 \tilde M \sigma_1$ satisfies  (\ref{eq:RHpair1}). 
Therefore, if (\ref{eq:RHpair1}) is uniquely solvable, so is   (\ref{eq:RHpair2}) and we have $\sigma_1 \tilde M \sigma_1=M$, 
hence 
\begin{equation}
  \tilde M=\sigma_1 M \sigma_1, \quad
  \tilde M_{21}=M_{12}, \quad \tilde M_{12}=M_{21}.
  \label{eq:m12m21}
\end{equation}

Now we come back to our specific situation. 
Recall that 
\begin{equation}
 \omega_\pm^{\mathrm{even}, \infty}
  =- \sigma_1 \bar \omega_\mp^{\mathrm{odd}, \infty} \sigma_1 .  
\label{eq:omegasigma1}
\end{equation}
We are almost in the situation described in (\ref{eq:RHpair1}) 
and (\ref{eq:RHpair2}). 
On the right-hand side of (\ref{eq:omegasigma1}), there is a 
negative sign  and the subscript $\mp$ replaces $\pm$. 
These deviations from \eqref{eq:RHpair2} are  canceled out by the fact that 
$\Sigma(0)_j$ is oriented differently 
in accordance  with the parity of $j$. See \eqref{eq:oddeven} below.

When $j$ is even,  we reverse the orientation of $\Sigma(0)_j$. 
Then  the orientation is now inward and the 
new jump matrix is 
\(
 v_{\raisebox{-1pt}{\mbox{\scriptsize{even}}}}
 =v^j(z)^{-1}
 =(I+\omega_+^{\mathrm{even}, \infty})^{-1}
   (I-\omega_-^{\mathrm{even}, \infty})
\).

Set 
\(
 v_\mathrm{odd}  (\bar r(S_j)) 
 = 
 (I-\bar\omega_-^{\mathrm{odd}, \infty})^{-1}
 (I+\bar\omega_+^{\mathrm{odd}, \infty})
\), then \eqref{eq:omegasigma1} implies
\begin{equation}
\label{eq:oddeven}
 v_{\mathrm{even}} =\sigma_1  v_\mathrm{odd}  (\bar r(S_j)) \sigma_1.
\end{equation}
Therefore the solution in Case A  with 
$r (S_j)$  and $\bar r (S_j)$ interchanged 
gives that in Case B by  the procedure in  (\ref{eq:m12m21}). 
If $j$ is even,    (\ref{eq:mj112}) and  (\ref{eq:m12m21}) lead to 
\begin{equation}
 (m^j_1)_{21} =
 \frac{ i \sqrt{2\pi}  e^{ i \pi/4}   e^{-\pi\nu_j/2}  }
            {\bar r(S_j) \Gamma(- i \nu_j)}.
 \label{eq:mj121even}
\end{equation}
It follows that 
\( (m^4_1)_{21} =- (m^2_1)_{21}\).

\begin{proposition}
\label{prop:1030B}
We have 
\begin{equation}
 R_n(t)
 =-\frac{\delta(0)}{\pi i}\sum_{j=1}^2 \beta_j (\delta_j^0)^{-2}
  S_j^{-2}
 (m^j_1)_{21}
  +O(t^{-1}\log t). 
 \label{eq:Rnt1030B}
\end{equation}
\end{proposition}
\begin{proof}
In (\ref{eq:Rnt1030}), the third and the fourth terms are 
identical with the first and the second respectively, 
because  
$\beta_{j+2}=-\beta_j$, $(\delta^0_{j+2})^2=(\delta^0_j)^2$, 
$S_{j+2}^2=S_j^2$ and  $(m^{j+2}_1)_{21}=- (m^j_1)_{21}$. 
Notice that $(m^j_1)_{21}$ is referred to as $M_j$ in 
Theorem~\ref{thm:main}.
\end{proof}

\subsection{Proof of Theorem~\ref{thm:main}}

\smallskip
Substitute (\ref{eq:mj112}) and (\ref{eq:mj121even}) into 
(\ref{eq:Rnt1030B}). Then we get Theorem~\ref{thm:main} 
in view of (\ref{eq:deltaj0defdecember}) and (\ref{eq:D1D2december}). 
See Remark~\ref{rem:rSj=0} for the case $r(S_j)=0$.

{Hideshi \sc{YAMANE}}
{Department of Mathematical Sciences,  Kwansei Gakuin 
University, \\Gakuen 2-1, Sanda, Hyogo 669-1337, Japan, \\
\texttt{yamane@kwansei.ac.jp}}

\label{finishpage}

\end{document}